\pdfoutput=1
\documentclass[11pt]{article}
\usepackage{lipsum}

\def\colorful{1}

\ifnum\colorful=1

\newcommand{\anote}[1]{\footnote{{\bf [Ankit: {#1} ]}}}
\newcommand{\hnote}[1]{\footnote{{\bf [Hadi: {#1} ]}}}
\newcommand{\vnote}[1]{\footnote{{\bf [Varun: {#1} ]}}}
\else

\newcommand{\anote}[1]{}
\newcommand{\hnote}[1]{}
\newcommand{\vnote}[1]{}
\fi

\usepackage[utf8]{inputenc} %
\usepackage[T1]{fontenc}    %
\usepackage{url}            %
\usepackage{booktabs}       %
\usepackage{nicefrac}       %
\usepackage{microtype}      %
\usepackage{amsthm,amsfonts,amsmath,amssymb,epsfig,color,float,graphicx,verbatim, enumitem}

\usepackage{algpseudocode,algorithm,algorithmicx}  

\usepackage{bbm}
\usepackage{caption}
\usepackage[
backend=biber,
style=alphabetic,
maxbibnames=15,
maxalphanames=10,
minalphanames=6,
doi=false,
isbn=false,
url=false,
eprint=false,
backref=true,
]{biblatex}
\usepackage[dvipsnames,table,xcdraw]{xcolor}

\usepackage[colorlinks,citecolor=RoyalBlue,linkcolor=RubineRed,bookmarks=true]{hyperref}

\usepackage[nameinlink,capitalise]{cleveref}

\usepackage{thm-restate}

\usepackage{mleftright}

\def\E{\mathbb E}

\newcommand{\poly}{\mathrm{poly}}

\newcommand{\cT}{\mathcal{T}}

\newcommand{\cY}{\mathcal{Y}}

\newcommand{\cH}{\mathcal{H}}

\newcommand{\cX}{\mathcal{X}}

\newcommand{\cB}{\mathcal{B}}

\newcommand{\cTT}{\mathcal{T}^\textup{\texttt{thresh}}}

\newcommand{\nstar}{n^{\!\ast}}

\newcommand{\alb}{\Bar{\pi}}
\newcommand{\lamb}{\Bar{\lambda}}
\newcommand{\pb}{\Bar{p}}
\newcommand{\qb}{\Bar{q}}
\newcommand{\betl}{\beta_{\lambda}}
\newcommand{\betls}{\beta_{\lambda_*}}

\newcommand{\ls}{\ensuremath{\lambda_*}}
\newcommand{\lsb}{\ensuremath{\Bar{\lambda_*}}}
\newcommand{\Ber}{\mathrm{Ber}}
\newcommand{\nsid}{n^*_{\mathrm{id}}}
\newcommand{\nsseq}{n^*_{\mathrm{seq}}}
\newcommand{\nsnid}{n^*_{\mathrm{non\text{-}id}}}
\newcommand{\prob}{\ensuremath{\mathbb{P}}}

\newcommand{\real}{\ensuremath{\mathbb{R}}}

\usepackage{cleveref, thmtools}
\crefformat{equation}{(#2#1#3)}
\crefname{equation}{Equation}{Equations}
\crefname{lem}{Lemma}{Lemmata}
\crefname{claim}{Claim}{Claims}
\crefname{fact}{Fact}{Facts}
\crefname{thr}{Theorem}{Theorems}
\crefname{prop}{Proposition}{Propositions}
\crefname{cor}{Corollary}{Corollaries}
\crefname{rem}{Remark}{Remarks}
\crefname{define}{Definition}{Definitions}
\crefname{question}{Question}{Questions}
\crefname{condition}{Condition}{Conditions}
\crefname{figure}{Figure}{Figures}
\newtheorem{thr}{Theorem}
\newtheorem{prop}{Proposition}
\newtheorem{lem}{Lemma}
\newtheorem{cor}{Corollary}[thr]
\newtheorem{fact}{Fact}
\newtheorem{question}{Question}

\newtheorem{define}{Definition}

\usepackage{enumitem}

\allowdisplaybreaks

\usepackage{color}
\definecolor{DGreen}{rgb}{0,0,0}

\def\dgreen{\color{DGreen}}

\renewcommand{\left}{\mleft}
\renewcommand{\right}{\mright}

\title{The Sample Complexity of Distributed Simple Binary Hypothesis Testing under Information Constraints}
\author{
Hadi Kazemi
\\
Department of Mathematics,\\ University of Cambridge\\
{\tt hk569@cam.ac.uk}
\and
Ankit Pensia\thanks{Supported by Research Pod on Resilience in Brain, Natural, and Algorithmic Systems at the Simons Institute.}\\
Department of Statistics and Data Science,\\ Carnegie Mellon University \\
{\tt ankitp@cmu.edu}
\and
Varun Jog\\
Department of Mathematics,\\ University of Cambridge\\
{\tt vj270@cam.ac.uk}
}

\begin{document}
\date{}
\maketitle

\begin{abstract}%
This paper resolves two open problems from a recent paper~\cite{PenEtal24b} concerning the sample complexity of distributed simple binary hypothesis testing under information constraints. The first open problem asks whether interaction reduces the sample complexity of distributed simple binary hypothesis testing. In this paper, we show that sequential interaction does not help. The second problem suggests tightening existing sample complexity bounds for communication-constrained simple binary hypothesis testing. We derive optimally tight bounds for this setting and resolve this problem. Our main technical contributions are: (i) a one-shot lower bound on the Bayes error in simple binary hypothesis testing that satisfies a crucial tensorisation property; (ii) a streamlined proof of the formula for the sample complexity of simple binary hypothesis testing without constraints, first established in~\cite{PenEtal24b}; and (iii) a reverse data-processing inequality for Hellinger-$\lambda$ divergences, generalising the results from \cite{BhaEtal21} and \cite{PenEtal23}.\footnote{A 1-page extended abstract of this paper, titled ``The Sample Complexity of Distributed Simple Binary Hypothesis Testing under Information Constraints,'' appeared in the Proceedings of the Conference on Learning Theory, 2025.}
\end{abstract}


\section{Introduction}
Hypothesis testing is one of the central problems in statistics that concerns testing different hypotheses based on observed data. The most basic version of this problem is simple binary hypothesis testing, where we wish to decide between two hypotheses $H_0$ and $H_1$. Under hypothesis $H_0$ (resp.\ $H_1$) we observe i.i.d.\ samples from distribution $Q$ (resp.\ $P$), where $P$ and $Q$ are known. The Bayes formulation of simple binary hypothesis testing, which is the focus of this paper, assumes that the true hypothesis is chosen randomly according to a prior Bernoulli distribution $\Ber(\pi)$. This is a classical problem that has been studied for over a century, and we highlight some important results.
\begin{itemize}
    \item \emph{Optimal algorithm:} The maximum a posteriori-rule (MAP-rule) minimises the error probability. In the prior-free setting, threshold tests are optimal \cite{Ney33}.
    \item \emph{Asymptotic error probability:} The error probability decays exponentially, with the exponents characterised by information-theoretic measures such as the Kullback--Leibler divergence or Chernoff information~\cite[Chapter~16]{powu2025}.
    \item \emph{Sample complexity:} The sample-size needed to achieve an error of $\delta$ was notably unaddressed until recently, when \cite{PenEtal24b} derived a formula for all $\delta \le \min(\pi, 1-\pi)/4$. 
\end{itemize}
In addition to identifying the sample complexity, \cite{PenEtal24b} also noted the implications of their result to the setting of distributed simple binary hypothesis testing under information constraints, such as communication (or quantisation) or local differential privacy. These had been studied in prior work \cite{PenEtal23, PenEtal24a} but only for the restricted setting of constant $\pi$ and $\delta$ (specifically, $\pi = 0.5$ and $\delta = 0.1$). This paper addresses two open problems noted in \cite{PenEtal24b}: tightening the existing sample complexity bounds for information-constrained hypothesis testing, and examining whether interaction helps reduce the sample complexity. In what follows, we describe the problem setting for distributed information-constrained hypothesis testing in more detail.
\paragraph{Hypothesis testing under information constraints:}
Suppose each of $n$ agents, denoted by $s_1, \dots, s_n$, receives a single sample denoted by $X_i$ for agent $s_i$. Agent $s_i$ transmits $Y_i$ to a central server, and the central server performs a hypothesis test using $Y_1, \dots, Y_n$. The transformation of $X_i$ to $Y_i$ is subject to information constraints that are captured by a set of channels $\cT$ (Markov kernels) from $\cX$, the support of the $X_i$s, to $\cY$, the support of the $Y_i$s.   i.e., the agent $s_i$ must pick a communication channel $T^i \in \cT$. For example, $\cT$ could be the set of channels with output size $D$, denoted by $\cT_D$, or the set of all $\epsilon$-differentially private channels, denoted by $\cT_{\epsilon\text{-LDP}}$. 

We assume agents have access to private randomness as well as public randomness. The assumption of private randomness is equivalent to assuming the convexity of $\cT$, since each agent can choose any convex mixture of the channels available to use. {\dgreen In contrast to other hypothesis-testing settings, such as the composite testing problems studied in~\cite{acharya2020b}, because we test between two prespecified distributions, it turns out that public randomness is of no use in our setting.} To see this, suppose all agents observe a public random variable $U$ which they may use to coordinate their choices of channels. The resulting algorithm (either in the identical, non-identical, or sequential settings described below) may be thought of as a mixture of private-randomness-only algorithms depending on the realisations of $U$. It is easy to see that if the Bayes error is at most $\delta$ after averaging over $U$, then there must be a particular realisation of $U$ that makes the error at most $\delta$. This shows that there exists a private-randomness-only algorithm that is as good as any public-and-private randomness algorithm. Hence, in what follows, we disregard the public and private random variables and focus on deterministic algorithms with a convex set $\cT$.

The different settings of interest based on how the agents select their channels are described in \cref{dfn:info_hyp}. 


\begin{define}[Information-Constrained Simple Binary Hypothesis Testing]\label{dfn:info_hyp}
Let $P$ and $Q$ be a pair of distributions on a discrete set $\cX$. Let $\pi\in(0,0.5]$ and $\Theta \sim \Ber(\pi)$. When $\Theta = 0$ (resp.\ $\Theta = 1$) we have $X_i\overset{i.i.d.}{\sim}Q$ (respectively $X_i\overset{i.i.d.}{\sim}P$). Let $\delta\in(0,1)$ and $\cT$ a convex set of channels with input space $\cX$ and output space $\cY$. Consider the channel and estimator pair $(T,\hat{\theta})$ where $\mathbf{T}=(T^1,T^2,\dots,T^n)$ is a sequence of $n$ channels where all channels have the same output space $\cY$, 
and $\hat{\theta}:\cY^n\rightarrow\{0,1\}$ is the tester. 
\begin{enumerate}

   \item 
    (Unconstrained) We denote the problem of simple binary hypothesis testing without constraints between two distributions $P$ and $Q$, with prior $\pi$ and error probability $\delta$ as $\cB_B(P,Q,\pi,\delta)$. The sample complexity of simple binary hypothesis testing, denoted by $n^*(P,Q,\pi,\delta)$, is the smallest number of samples $n$ for which there exists a test $\hat \theta$ satisfying 
    \begin{equation}\label{eq:hp_def_1}
        \pi \cdot \prob_{X\sim P^{\otimes n}}[\hat{\theta}(X)\neq 1]+\alb \cdot \prob_{X\sim Q^{\otimes n}}[\hat{\theta}(X)\neq 0]\leq \delta,
    \end{equation}
    where we use the shorthand $X = (X_1, \dots, X_n)$.
    \item (Identical) 
    $\mathbf{T}=(T,T,\dots,T)$ where $T\in\cT$. We say that the pair $(T,\hat{\theta})$ solves the information constraint hypothesis testing problem in identical channel setting with $n$ samples with prior $\pi$ and error probability $\delta$ if 
    \begin{equation}\label{eq:hp_def}
        \pi \cdot \prob_{X\sim P^{\otimes n}}[\hat{\theta}(Y)\neq 1]+\alb \cdot \prob_{X\sim Q^{\otimes n}}[\hat{\theta}(Y)\neq 0]\leq \delta\,.
    \end{equation}
    In the above, we denote $(X_1,X_2,\dots,X_n)$ by $X$ and $(Y_1,Y_2,\dots,Y_n)$ by $Y$. Also, we assume that we have $Y_i=T(X_i)$. The sample complexity $\nsid(P,Q,\pi,\delta,\cT)$ of information-constrained hypothesis testing with identical channels is the smallest $n$ for which there exists a pair of channel and tester $(T,\hat{\theta})$ that solves the problem with $n$ samples. 
        \item (Non-identical and non-interactive)
        $\mathbf{T}=(T^1,T^2,\dots,T^n)$ where $T^i\in\cT$. We say that the pair $(\mathbf{T},\hat{\theta})$ solves the information-constrained hypothesis testing with non-identical (and non-interactive) channels if they satisfy \cref{eq:hp_def} where $Y_i=T^i(X_i)$. The sample complexity is denoted by $\nsnid(P,Q,\pi,\delta,\cT)$.
    \item (Sequentially interactive) 
    $\mathbf{T}=(T^1,T^2,\dots,T^n)$ and each $T^i$ is a channel from $\cX\times\cY^{i-1}$ to $\cY$, such that for any fixed $(y_1,y_2,\dots,y_{i-1})$, the stochastic map $(X_i,y_1,y_2,\dots,y_{i-1}) \mapsto Y_i$ is in $\cT$.  
    We say that the pair $(\mathbf{T},\hat{\theta})$ solves this problem if they satisfy \cref{eq:hp_def} where $Y_i=T^i(X_i,Y_1,Y_2,\dots,Y_{i-1})$. We denote the sample complexity of this problem by $\nsseq(P,Q,\pi,\delta,\cT)$.
\end{enumerate}
\end{define}
When $P$, $Q$, $\pi$, $\delta$, and $\cT$ are clear from the context, we denote the sample complexities by  $n^*$, $\nsid$, $\nsnid$ and $\nsseq$.
Naturally, we must have
\begin{equation*}
\nstar \le 
        \nsseq
        \le
        \nsnid
        \le
        \nsid.
   \end{equation*}   
The first question we address concerns the relationship between the sample complexities in the three settings described above. This question was noted as an open problem in \cite{PenEtal24b}.
\begin{question}\label{Question: seq-id}
Is the sample complexity with sequential interaction significantly smaller (i.e. smaller by more than a constant factor) compared to the sample complexity with identical channels? Adding the non-identical setting to the question, we consider if any of the inequalities are truly strict:
    \begin{equation}
        \nsseq(P,Q,\pi,\delta,\cT)
        \ll
        \nsnid(P,Q,\pi,\delta,\cT)
        \ll
        \nsid(P,Q,\pi,\delta,\cT).
    \end{equation}
\end{question}
This question is motivated by the fact that in many important settings, there is a significant difference between $\nstar$ and $\nsid$ (for example, communication constraints~\cite{PenEtal23} and local differential privacy~\cite{DucJW18}).
Hence, it is an important question whether a sequentially interactive choice of channels could lead to significant sample savings (i.e., whether $\nsseq \ll \nsid$ holds). 
Indeed, in many inference problems such as those explored in~\cite{AchEtal22}, sequential interactivity\footnote{There also exist testing problems where the choice of non-identical, non-interactive channels significantly lowers the sample complexity, i.e, $\nsnid\ll \nsid$~\cite[Remark 5.15]{PenEtal23}.} provably reduces the sample complexity~\cite{AchCLST22-interactive,PouAA24}. \Cref{Question: seq-id} asks whether it is also for true for \emph{simple binary hypothesis testing}. 

Towards answering this question, the relationship between $\nsnid$ and $\nsid$ was resolved in a very restricted setting (uniform prior, i.e, $\pi=1/2$) in \cite{PenEtal23}, where it was established that $\nsnid \asymp \nsid$ for all convex set of channels $\cT$. In this paper, we extend this result to almost all ranges of $\pi$ and $\delta$ . In fact, we show a stronger result, that {\dgreen $\nsseq \geq \frac{1}{4}(\nsid-1)$}, thus establishing that sequential interaction cannot give more than a constant factor improvement over the identical channels setting, answering \Cref{Question: seq-id}; see \Cref{id-nonid-channels}.\footnote{That is, while sequential interactivity might or might not help for some testing problems as established in \cite{AchCLST22-interactive,PouAA24}, it provably does not help for the problem of simple binary hypothesis testing.}

Next, we consider the special case of $\cT = \cT_D$; i.e., we look at the sample complexity of simple binary hypothesis testing under communication constraints where each channel has an output alphabet of size $D$.
Letting $\nsid$ be the sample complexity under this constraint (based on our first result, it does not matter whether the channels are chosen identically or in a sequentially-interactive manner).
We study the relationship between $\nsid$ and $n^*$, which is the classical problem of decentralised detection surveyed in~\cite{tsitsiklis1988decentralized}. 
\begin{question}\label{Question:nidcommun}
How does the communication budget $D$ affect the sample complexity of hypothesis testing for a general prior? In particular, is $\nsid \lesssim n^* \max \left( 1, \frac{\log n^*}{D}\right)$, and is it tight in the worst case?
\end{question}
For reasons of clarity, we assume $\delta = \pi/16$ (or any fixed constant fraction of $\pi$) in the following discussion. The relationship between $\nsid$ and $n^*$ described in Question 2 was first proved in \cite{PenEtal23} in a \emph{restricted parameter regime} ($\pi = 0.5$), where they established both the upper bounds and the worst-case lower bounds; see also \cite{BhaEtal21}. For general prior $\pi$,  \cite{PenEtal24b} extended this result by establishing an upper bound of $n^* \max\left(1, \frac{ \log (n^*/\pi)}{D}\right)$ and conjectured that the $1/\pi$ factor inside the logarithm should not be there.
With this context, \Cref{Question:nidcommun} asks whether the ``cost'' of communication constraints increases (or even decreases) as the prior becomes more skewed.
In this paper, we answer this question by showing that  $\nsid \lesssim n^* \max \left( 1, \frac{\log n^*}{D}\right)$ for almost the entire range of interest of the parameters, specifically, for any $\pi \le 1/2$ and any $\delta \le \pi/16$. Indeed, we show a slightly stronger upper bound that is complemented by a matching lower bound obtained by constructing a worst-case pair of distributions.
see \Cref{thm: sc_D}.

To answer these two questions, we develop two main technical results that might be of independent interest: (i) a one-shot lower bound on the error probability of simple binary hypothesis testing (\Cref{thm:one-shot}); and (ii) and a sharp reverse data-processing inequality for the Hellinger-$\lambda$ divergences (\Cref{revdatainf}). A notable application of (i) is that it yields a short and neat proof of the formula for $n^*(P, Q, \pi, \delta)$, whose discovery was the primary contribution in \cite{PenEtal24b}; see \Cref{samplecomplxsimplecor}.

\subsection{Preliminaries}
We provide some basic definitions and facts about hypothesis testing and $f$-divergences that are needed to discuss our results.
\begin{define}[$f$-divergence \cite{powu2025}]\label{dfn: f-div}
    Let $f:[0,\infty]\rightarrow\real\cup\{\infty\}$ be a convex function with $f(1)=0$. For two distributions $P$ and $Q$ we define $D_f(\mu\parallel\nu)$ to be
    \begin{equation}
        D_f(P\parallel Q)=\E_{Q}\left[f\left(\frac{dP}{dQ}(Z)\right)\mathbbm{1}_{\{\frac{dP}{dQ}(Z)<\infty\}}\right]+f'(\infty)P\left(\frac{dP}{dQ}(Z)=\infty\right).
    \end{equation}
\end{define}
\begin{fact}[Convexity of $f$-divergences]\label{fact:fdiv-cvx}
    All f-divergences are jointly convex in both inputs.
\end{fact}
\begin{define}[Hellinger-$\lambda$ divergence]
    For any $\lambda\in(0,1)$ we define the Hellinger-$\lambda$ divergence between two distributions $P$ and $Q$ as $H_{\lambda}(P,Q)=1-\E_Q\left[\left(\frac{dP}{dQ}\right)^\lambda\right]$. We also define $\betl(P,Q)=\E_Q\left[\left(\frac{dP}{dQ}\right)^\lambda\right]=1-H_{\lambda}(P,Q)$ to be the Hellinger-$\lambda$ affinity. Hellinger-$\lambda$ is an $f$-divergence with $f(x)=1-x^\lambda$.
\end{define}

\begin{fact}[Tensorisation of Hellinger affinity]
 Hellinger affinity tensorises; i.e., for every $\lambda\in(0,1)$ and any $P$, $Q$ any $n\in\mathbb{N}$ we have
        $\betl(P^{\otimes n},Q^{\otimes n})=\left(\betl(P,Q)\right)^n.$
\end{fact}
\begin{fact}[Inequality between Hellinger-$\lambda$ divergences \cite{bar2002complexity}]\label{Hlambdaupper}
    For any pair of distributions $\mu$ and $\nu$ and
    $\forall\ 0<\alpha<\beta<1$, we have
       $ \frac{\alpha}{\beta}H_{\beta}(\mu,\nu)
        \leq
        H_{\alpha}(\mu,\nu)
        \leq
        \frac{1-\alpha}{1-\beta}H_{\beta}(\mu,\nu)$.
\end{fact}

\begin{fact}[Optimal Bayes error]\label{fact: bayes}
Consider $\cB_B(P,Q,\pi,\delta)$ as in Definition~\ref{dfn:info_hyp}. Let $e_\pi(P,Q)$ denote the minimum error probability over all estimators $\hat{\theta}:\cX\rightarrow\{0,1\}$ based on a single sample $X$. The following facts are folklore~\cite{HelRav70}:

\begin{enumerate}
\item
The error of the optimal estimation is
       $e_\pi(P,Q)=\sum_{x\in \cX}\min\{\pi P(x),\alb Q(x)\}.$
    
\item
For any $\lambda\in [0,1]$ we have
       $e_\pi(P,Q)\leq \pi^\lambda {\alb}^{\lamb}\betl(P,Q).$
\end{enumerate}
\end{fact}
\paragraph{Notation:}
For two real-valued non-negative functions on an arbitrary space $\mathcal{X}$, i.e. $f,g:\mathcal{X}\rightarrow\mathbb{R}$, we use $f(x)\asymp g(x)$ to indicate that there exists universal constants $C_1,C_2>0$ s.t. $\ C_1 g(x) \leq f(x) \leq C_2 g(x)$ for all $x\in \cX$. We use $f(x)\lesssim g(x)$ to indicate that there is a universal constant $C>0$ such that $f(x)\leq C g(x)$ for all $x\in \cX$; The relation $f(x) \gtrsim g(x)$ is defined analogously.
For $t\in [0,1]$ we define $\Bar{t}:=1-t$. 

\subsection{Main Results}
Our first result is a lower bound on $e_\pi(P,Q)$ from \cref{fact: bayes}, stated informally below: 
\begin{thr}[One-shot lower bound on the Bayes error (informal)]
    \label{thm:one-shot}
    Let $P,Q$ be a pair of distributions and $\pi\in(0,\frac{1}{2}]$. Also, let $e_\pi(P,Q)=\delta$, which is assumed to satisfy $\delta \in (0,\pi)$. Then there exists $\ls=\ls(\pi,\delta)\in [0.5,1)$ such that
    \begin{equation}\label{intr:ealphalower}
         \delta
         \geq 
         \frac{1}{4}\pi^{\ls} \alb^{\lsb} \left(\betls (P,Q)\right)^2.
    \end{equation}
\end{thr}
The inequality in \cref{intr:ealphalower} shows that the classical upper bound in \cref{fact: bayes} is tight in a certain sense, when $\lambda$ is chosen carefully. 
The crucial advantage of this bound compared to all other existing bounds on $e_\pi(P,Q)$ in the literature, such as Fano's bound, is that the right hand side tensorises. This makes it an ideal tool for deriving lower bounds on the sample complexity.  A particular application to $P^{\otimes n}$ and $Q^{\otimes n}$ leads to our next result, \cref{samplecomplxsimplecor}, which recovers the exact bounds for the sample complexity of simple binary hypothesis testing from \cite{PenEtal24b}.

\begin{restatable}[Sample complexity of simple binary hypothesis testing]{thr}{samplecomplxsimplecor}\label{samplecomplxsimplecor}
Consider $\cB_B(P,Q,\pi,\delta)$ as in \cref{dfn:info_hyp} where $\pi\in(0,0.5]$ and $\delta\in(0,\frac{1}{16}\pi]$. For $\ls=\frac{\log(\frac{\alb}{\delta})}{\log(\frac{\alb}{\delta})+\log(\frac{\pi}{\delta})} \in [0.5,1)$, we have
\begin{equation}\label{samplecomplxsimple}
    \left\lceil\frac{1}{2}\ls\frac{\log(\frac{\pi}{\delta})}{\log\left(\frac{1}{\betls(P,Q)}\right)}\right\rceil
    \leq
    n^*(P,Q,\pi,\delta)
    \leq
    \left\lceil 2\ls\frac{\log(\frac{\pi}{\delta})}{\log\left(\frac{1}{\betls(P,Q)}\right)}\right\rceil.
\end{equation}
\end{restatable}

The above theorem reproduces all the results about the sample complexity of simple binary hypothesis testing in \cite{PenEtal24b} in one neat formula. The proof of this result in \cite{PenEtal24b} consists of three separate proofs, with entirely different strategies: (i) when $\delta \le  \pi^2$, the proof uses inequalities between the Hellinger-$1/2$ divergence and the total variation (TV) distance; (ii) when $\delta \in (\pi^2, \pi/100)$, it relies on a novel argument based on error and accuracy amplification; and (iii) when $\delta \in (\pi/100, \pi/4]$ the proof requires a technical $f$-divergence inequality between Jensen--Shannon $f$-divergences and Hellinger-$\lambda$ divergences. The main advantage of our result is that it does away with the three (somewhat artificial) regimes and leads to a formula in a simple and straightforward manner. We also note that the constants in the lower and upper bounds in \cref{samplecomplxsimplecor} are within a factor of 4, leading to a tighter and potentially practically useful characterisation of the true sample complexity.


Our next result is another application of our one-shot lower bound on the Bayes error, which gives a complete answer to \cref{Question: seq-id}.
\begin{restatable}[Sequential interactivity does not help]{thr}{idnonidchannels}\label{id-nonid-channels}
  Let $\cT$ be any convex set of channels. For an arbitrary pair of distributions $P$ and $Q$ and a set of channels $\cT$, let $\pi\in(0,0.5]$ and $\delta\in(0,\frac{1}{16}\pi]$. Then 
        $\nsseq
        \geq
        \frac{1}{4}(\nsid-1).$
\end{restatable}
Note that from the definition of $\nsseq$ and $\nsid$ it is immediate that we have $\nsseq\leq \nsid$. {\dgreen Equivalently, $\nsid\leq 4\nsseq+1$, so sequential interaction cannot yield more than a constant-factor improvement, apart from the additive-one term.} We also remark that \Cref{thm:one-shot} was crucial in proving this result. Indeed, adapting the proof technique from \cite{PenEtal24b} fails for non-identical channels (and thus, also for sequentially interactive channels) because the error amplification trick used when $\delta \in (\pi^2,\pi/100)$ in \cite{PenEtal24b} only works when the channels are identical.


Next, we turn our attention to the problem of hypothesis testing under communication constraints. In this setting, we assume that $\cT=\cT_D$, where $\cT_D$ is the set of all channels with output space $\cY$ such that $\lvert\cY\rvert=D$ for some $D\geq 2$. From \cref{samplecomplxsimplecor} we know that the sample complexity of hypothesis testing characterised by the Hellinger-$\lambda$ divergence between them. So the sample complexity under $\cT_D$ depends on the largest the Hellinger divergence between the two output distributions after the quantisation step. This problem was studied for the Hellinger-$1/2$ divergence in \cite{PenEtal23} via a ``reverse data-processing inequality''. Our main technical contribution is significantly generalising this inequality to a large class of ``TV-like $f$-divergences'' (defined formally in \cref{def: tv_f}), which include all Hellinger-$\lambda$ divergences.

\begin{thr}[Reverse data-processing inequality for TV-like divergences (informal)]\label{revdatainf}
    Let $P$ and $Q$ be two discrete distributions, and let $D_f$ be a TV-like f-divergence. Then for any $D\geq 2$ there exists a threshold quantiser $T^*\in \cT_D$ such that the following inequality holds \begin{equation}
        D_f(T^*P\parallel T^*Q)
        \gtrsim
        \min\left\{1,\frac{D}{\log(\frac{1}{D_f(P\parallel Q)})}\right\}D_f(P\parallel Q).
    \end{equation}
\end{thr}

The above result shows that when $D_f(P \parallel Q)$ is small enough and $D=2$, we can preserve the original $f$-divergence up to logarithmic factors. When $D$ becomes sufficiently large; i.e, $D \gtrsim \log\left(\frac{1}{D_f(P \parallel Q)}\right)$, we can preserver the original $f$-divergence up to constant factors.

We can also interpret \cref{revdatainf} as a tightened analysis of the binary information distillation function which was studied in \cite{BhaEtal21}. More precisely, let $X$ be a binary random variable and let $[Y]_M$ be $Y$ quantised to $M$ levels. \cite{BhaEtal21} analysed the function 
\begin{align*}
  \mathrm{ID}_M(2,\beta)\triangleq\sup_{P_{XY}: I(X;Y)\geq\beta} I(X,[Y]_M)  
\end{align*}
and showed that $\mathrm{ID}_M(2,\beta)\gtrsim\min\{\frac{M}{\log(\frac{1}{\beta})},1\}\beta$. This result was used in \cite{PenEtal24b} to derive (loose) achievable bounds on $\nsid$. 
We can tighten the result in \cite{BhaEtal21} to incorporate additional information about $X$ using \cref{revdatainf} to state  
\begin{align*}
  \mathrm{ID}_M(2,\beta,\pi)\triangleq\sup_{\substack{P_{XY}: I(X;Y)\geq\beta\\
  X \sim \Ber(\pi)}} I(X,[Y]_M)\quad  \gtrsim\min\left\{\frac{M}{\log\frac{h(\pi)}{\beta}},1\right\}\beta\,,
\end{align*}
    where $h(.)$ is the binary entropy function. As the bias of $X$ goes to $0$, our result improves upon  \cite{BhaEtal21} by replacing $\log(1/\beta)$ in the denominator with $\log\frac{h(\pi)}{\beta}$.
    This tightened analysis relies on a technical result in \cite{PenEtal24b} that shows an equivalence between mutual information and the Hellinger-$\lambda$ divergence for a suitable choice of $\lambda$. 

Applying \cref{revdatainf} for Hellinger-$\lambda$ divergences, we answer \cref{Question:nidcommun} affirmatively below.

\begin{restatable}[Sample complexity under $\cT_D$ constraint]{thr}{samplecompcommunication}\label{thm: sc_D}
    Let $P$ and $Q$ be two discrete probability distributions on $\cX$ and $D\geq 2$. Assume that $H_{1/2}(P,Q)\leq 0.25$. Then for every $\pi\in(0,0.5]$ and $\delta\in(0,\frac{1}{16}\pi]$, if we denote $n^*(P,Q,\pi,\delta)$ by $n^*$,  we have
    \begin{equation}
        \nsseq(P,Q,\pi,\delta,\cT_D)
        \lesssim
        \max\left\{1,\frac{\log\left(\frac{1}{H_{\ls}(P,Q)}\right)}{D}\right\}n^*
        \lesssim
        \max\left\{1,\frac{\log(n^*)}{D}\right\}n^*.
    \end{equation}
    Moreover, for every $\pi$ and $\delta$ satisfying the above conditions, there exist $P$ and $Q$ such that 
    \begin{equation}
        \nsseq(P,Q,\pi,\delta,\cT_D)
        \gtrsim
        \max\left\{1,\frac{\log(\frac{1}{H_{\ls}(P,Q)})}{D}\right\}n^*.
    \end{equation}
\end{restatable}
When $D$ is small, this result implies that the sample complexity with constraints is at most $n^*\log n^*$; i.e., communication constraints only have a mild impact by increasing the sample complexity by a logarithmic factor compared to the unconstrained sample complexity. This result generalises a similar result proved in \cite{PenEtal24b} for the restricted setting of constant $\pi$.

Finally, our bounds also imply computational benefits under $\epsilon$-local differential privacy constraints. Specifically, restricting $\epsilon$-LDP channels to $D$ outputs only has a mild impact on the sample complexity with $\epsilon$-LDP constraints, and furthermore, the optimal channel with $D$ outputs can be found efficiently. {\dgreen The algorithmic result of~\cite[Theorem~2.6]{PenEtal24b}, combined with our reverse data-processing inequality in \Cref{revdata}, gives the following corollary.}

\begin{cor}[Computational benefits for private hypothesis testing]
Let $\cT_{\epsilon\text{-LDP}}$ be the set of all $\epsilon$-LDP channels and $\cT_{\epsilon\text{-LDP}, D}$ be all $\epsilon$-LDP channels with output alphabet $D \ge 2$. Let $n^*_{\mathrm{priv}} \triangleq n^*(P, Q, \pi, \delta, \epsilon)$ be the sample complexity under  $\cT_{\epsilon\text{-LDP}}$ constraints, and $n^*_{\mathrm{priv, D}}\triangleq n^*(P, Q, \pi, \delta, \epsilon, D)$ under $\cT_{\epsilon\text{-LDP}, D}$ constraints. Then there exists an algorithm that runs in time polynomial in $|\cX|^{D^2}$ and $2^{D^3\log D}$ and finds a channel $T$ from $\cT_{\epsilon\text{-LDP}, D}$ such that the sample complexity with $T$ is at most $n^*_{\mathrm{priv}}\cdot\max\left(1,\frac{\log n^*_{\mathrm{priv}}}{D} \right)$.\footnote{Observe that in the important special case when $D$ is a small constant, the algorithm is polynomial in the support size of the distributions.}
\end{cor}

\section{Related Work}
The problem of hypothesis testing under communication constraints has been studied since the 1980s \cite{tsitsiklis1988decentralized}, largely in the asymptotic setting when $n$ goes to infinity. {\dgreen Classical work on hypothesis testing under communication constraints includes Ahlswede and Csisz\'ar~\cite{AhlswedeCsiszar1986}, while extremal properties of likelihood-ratio quantizers were studied by Tsitsiklis~\cite{Tsitsiklis1993}. Distributed detection with multiple sensors was surveyed by Viswanathan and Varshney~\cite{ViswanathanVarshney1997}.} There has been significant progress on the non-asymptotic and minimax setting for closely related problems of distribution estimation or identity testing~\cite{acharya2020a,acharya2020b,AchCLST22-interactive,shamir2014fundamental} and parametric estimation~\cite{ozgur2018geometric}. The works most closely related to this paper are \cite{PenEtal23,PenEtal24a} which study the sample complexity of simple binary hypothesis testing under communication and privacy constraints, respectively. As noted earlier, this paper is motivated by the open problems from~\cite{PenEtal24b}, which primarily studied the unconstrained setting. 
Related to \cref{thm: e_alpha_lb}, we note that the information theory literature is rich with various lower bounds on the Bayes error due to the natural connections between $M$-ary hypothesis testing and channel coding. See, for example~\cite{VazEtal16, Ren66, Sha11, SasVer17, FedMer94, BenRav78, Ito72, CovHar67, Dev74, Kai67, Tou72, Vaj68, SanVar06}. Unfortunately, \emph{none} of these bounds are suitable for studying the sample complexity because they do not \emph{tensorise}; i.e., the error probability bound for $n$-fold product distributions cannot be evaluated easily. The information theory literature also contains non-asymptotic bounds on the Type-I and Type-II error probabilities (see \cite{LunKon24} for an exposition), however, it is unclear if these bounds can lead to tight sample complexity bounds for binary hypothesis testing. We also note that reverse data-processing bound from \cref{revdata} is closely related to the problem of quantisation while preserving information, with some related works being~\cite{BhaEtal21,pedarsani2011,tal2012,kurkoski2014}.

\section{One-Shot Tight Lower Bound on the Error of Simple Binary Hypothesis Testing}\label{one-shot}

Motivated by the study of the sample complexity of simple binary hypothesis testing with sequential interaction and information constraints, in this section, we prove a one-shot lower bound on the error probability of simple binary hypothesis testing problem. This one-shot lower bound will give us a powerful tool to determine the sample complexity of binary hypothesis testing under any set of information constraints. 

\setcounter{thr}{0}
\begin{thr}[One-shot lower bound on the Bayes error]
    \label{thm: e_alpha_lb}
    Let $P,Q$ be a pair of distributions and $\pi\in(0,\frac{1}{2}]$ such that $e_{\pi}(P,Q)\in(0,\pi)$. Then for any {\dgreen $\delta\in[e_{\pi}(P,Q),\pi)$} there exists $\ls=\ls(\pi,\delta)\in [0.5,1)$ such that
    \begin{equation}\label{ealphalower}
         \delta
         \geq 
         \frac{1}{4}\pi^{\ls} \alb^{\lsb} \left(\betls (P,Q)\right)^2,
    \end{equation}
    where 
        $\ls(\pi,\delta)
        =
        \frac{\log(\frac{\alb}{\delta})}{\log(\frac{\alb}{\delta})+\log(\frac{\pi}{\delta})}.$
    In particular, when $\delta=e_{\pi}(P,Q)$, there exists $\ls$ (that depends on $e_{\pi}(P,Q)$) such that 
        $e_{\pi}(P,Q)
        \geq
        \frac{1}{4}
        \pi^{\ls}\alb^{\lsb}\left(\betls(P,Q)\right)^2.$
\end{thr}
\setcounter{thr}{6}

\begin{proof}[Proof of \cref{thm: e_alpha_lb}]
    First, suppose that $e_\pi(P,Q)=\delta$. We will prove the case where $e_\pi(P,Q)<\delta$ later. Our first lemma reduces the problem to Bernoulli $P$ and $Q$. For brevity, we use $e_\pi(p,q)$ instead of $e_\pi(\Ber(p), \Ber(q))$ in what follows, and similarly for $\beta_\lambda(\cdot, \cdot)$.
    \begin{restatable}[Reduction to Bernoulli]{lem}{binaryred}\label{lem: binary red}
        Assume for every pair of $(p,q)\in[0,1]^2$ we have the following
            $e_{\pi}(p,q)
            \geq
            \frac{1}{4}\pi^{\ls}\alb^{\lsb}\left(\betls\left(p,q\right)\right)^2,$
        where $\ls$ is chosen as in \cref{thm: e_alpha_lb}.
        Then for any $P$ and $Q$ on a discrete set $\cX$, we have
            $e_{\pi}(P,Q)
            \geq
            \frac{1}{4}
            \pi^{\ls}\alb^{\lsb}\left(\betls(P,Q)\right)^2.$
    \end{restatable}
    \begin{proof}[Proof of~\cref{lem: binary red}]
        Consider the set $A\triangleq\{x\in\cX|\pi P(x)\leq \alb Q(x)\}$ and define $p_A\triangleq P(A)$ and $q_A\triangleq Q(A)$. One can check that $e_\pi(P,Q)=e_\pi(\Ber(p_A),\Ber(q_A))$. We can write
    \begin{align}
        e_\pi(P,Q)
        =
        e_\pi(p_A, q_A)
        &\stackrel{(i)}{\geq}
        \frac{1}{4}\pi^{\ls}
    \alb^{\lsb}\left(\betls(p_A,q_A)\right)^2   \stackrel{(ii)}{\geq}
        \frac{1}{4}\pi^{\ls
        }\alb^{\lsb}\left(\betls(P,Q)\right)^2
    \end{align}
    where $(i)$ follows from the assumption we made in \cref{lem: binary red} and $(ii)$ follows from the data-processing inequality, having in mind that $H_{\ls}(P,Q)=1-\betls(P,Q)$ is a f-divergence.
    \end{proof}

The proof of Bernoulli case relies on the following technical lemma. 
 
    \begin{restatable}[Affinity upper bound for Bernoulli distributions]{lem}{techupbound}\label{techupbound}
        For $\ls$ as in \cref{thm: e_alpha_lb} and $\pi\in(0,1/2]$, we have         
        \begin{equation}             
        \sup_{\substack{(p,q)\in [0,1]^2:\\ e_\pi(p,q)=\delta}}\betls(p,q)
        \leq             
        2\left(\frac{\delta}{\pi}\right)^{\ls}.
        \end{equation}
        Moreover, we have $\left(\frac{\delta}{\pi}\right)^{\ls}=(\frac{\delta}{\alb})^{\lsb}$.    
    \end{restatable}
  
\begin{proof}[Proof Sketch of \cref{techupbound}] The proof of this lemma is in \Cref{app: affinity-upperbound}, but we give a brief proof sketch. From \cref{fact:fdiv-cvx} we know that the function $\betl(p,q)$ is jointly concave in its inputs. One can also see that the set of $(p,q)\in[0,1]^2$ that have $e_{\pi}(p,q)=\delta$ is the union of two line segments. The choice of $\ls$ is such that $\betls(p,q)$ is almost constant on the line segments, making it possible to bound the supremum by the value at the endpoints, which are easy to evaluate.
\end{proof}    
Using \cref{techupbound}, we can complete the proof of \cref{thm: e_alpha_lb} as follows:
    \begin{align}
     \frac{1}{4}\pi^{\ls} \Bar{\pi}^{\lsb} \left(\beta_{\ls}(p,q)\right)^2 
     \leq &
     \frac{1}{4}\pi^{\ls} \Bar{\pi}^{\lsb} \left(\sup_{\substack{(p,q)\in [0,1]^2:\\ e_\pi(p,q)=\delta}}\betls(p,q)\right)^2 
     \nonumber\\ 
     &\stackrel{(i)}\leq 
     \pi^{\ls} \Bar{\pi}^{\lsb} \left(\frac{\delta}{\pi}\right)^{2\ls}
     =
     \delta,
     \end{align}
where $(i)$ follows by \cref{techupbound}. This completes the proof of \cref{thm: e_alpha_lb} when $\delta=e_\pi(P,Q)$. When $\delta>e_\pi(P,Q)$, for some $a \notin \cX$ set $\tilde{P}(x) \triangleq \gamma P(x) \mathbbm 1_{x \neq a} + \bar \gamma \mathbbm 1_{x = a} $ and $\Tilde{Q}(x) = \gamma Q(x) \mathbbm 1_{x \neq a} + \bar \gamma \mathbbm 1_{x = a}$.
    Where $a \notin \cX$ is a new symbol. We set $\gamma$ such that $e_\pi(\Tilde{P},\Tilde{Q})=\delta$. Such a $\gamma$ always exists, since $e_\pi(\Tilde{P},\Tilde{Q}) = \gamma e_{\pi}(P,Q)+\Bar{\gamma}\pi$.
 We also have the following lower bound on $\betls(\Tilde{P},\Tilde{Q})$:
        \begin{align}\label{betlieq}
    \betls(\Tilde{P},\Tilde{Q})
    =&
    \sum_{x\in \cX\cup\{a\}}\Tilde{P}(x)^{\ls} \Tilde{Q}(x)^{\lsb}
    \gamma\betls(P,Q)+\Bar{\gamma}
    \overset{(i)}{\geq}
    \left(\betls(P,Q)\right)^\gamma 
    \overset{(ii)}{\geq}
    \betls(P,Q).
    \end{align}
    Where $(i)$ follows from the AM-GM inequality and $(ii)$ follows from the fact that $\betl(P,Q)\leq 1$ and $\gamma \leq 1$. Now we can use the inequality \cref{betlieq} above to prove the theorem for the case $\delta>e_\pi(P,Q)$ as follows:
    \begin{align*}
    \delta 
    =
    e_{\pi}(\Tilde{P},\Tilde{Q})
    \geq
    \frac{1}{4}\pi^{\ls} \Bar{\pi}^{\lsb} \left(\betls(\Tilde{P},\Tilde{Q})\right)^2
    \geq
    \frac{1}{4}\pi^{\lambda} \Bar{\pi}^{\lsb} \left(\betls(P,Q)\right)^2.
    \end{align*}
\end{proof}
\subsection{Sample Complexity of Unconstrained Simple Binary Hypothesis Testing}\label{sec:samplecompsimple}
In this section we derive a shorter and simpler proof of the sample complexity of simple binary hypothesis testing from \cite{PenEtal24b} using \cref{thm: e_alpha_lb}.  
\samplecomplxsimplecor*

\begin{proof}[Proof of \cref{samplecomplxsimplecor}]
    The upper bound in \Cref{samplecomplxsimple} follows simply by setting $n=\left\lceil 2\ls\frac{\log(\frac{\pi}{\delta})}{\log\left(\frac{1}{\betls(P,Q)}\right)}\right\rceil$ and using the upper bound on error probability in \cref{fact: bayes} as below
    \begin{align*}
        e_{\pi}(P^{\otimes n},Q^{\otimes n})
        \leq&
        \pi^{\ls}\alb^{\lsb}\beta_{\ls}(P^{\otimes n},Q^{\otimes n})
        =
        \pi^{\ls}\alb^{\lsb}\left(\beta_{\ls}(P,Q)\right)^n
        \le
        \pi^{\ls}\alb^{\lsb}(\frac{\delta}{\pi})^{2\ls}
        =
        \delta,
    \end{align*}
    where we used the fact that we have $\frac{\pi^{\ls}\alb^{\lsb}}{\delta}=\left(\frac{\pi}{\delta}\right)^{\ls}\left(\frac{\alb}{\delta}\right)^{\lsb}=\left(\frac{\pi}{\delta}\right)^{2\ls}$
    for the specific choice of $\ls$ in \cref{samplecomplxsimplecor}. 
    To prove the lower bound, suppose that for $\delta\in(0,\frac{1}{16}\pi]$ and some $n\in\mathbb{N}$ we have $e_\pi(P^{\otimes n},Q^{\otimes n})\leq \delta$. Then we use the lower bound we proved in \cref{thm: e_alpha_lb} to write
    \begin{align*}
        &\delta
        \geq 
        \frac{1}{4}\pi^{\ls} \alb^{\lsb}\left(\betls(P^{\otimes n},Q^{\otimes n})\right)^2 
        =
        \frac{1}{4}\pi^{\ls} \alb^{\lsb}\left(\betls(P,Q)\right)^{2n}
        \overset{(i)}{\geq}
        \left(\frac{\delta}{\pi}\right)^{\ls}\pi^{\ls} \alb^{\lsb}\left(\betls(P,Q)\right)^{2n}, 
    \end{align*}
    In, $(i)$ we used $\left(\frac{\pi}{\delta}\right)^{\ls}
    \geq
    \left(\frac{\pi}{\delta}\right)^{0.5}
    \geq
    4$. Rearranging and simplifying yields $n
        \geq
        \frac{\ls}{2} \frac{\log(\frac{\pi}{\delta})}{\log(\frac{1}{\betls(P,Q)})}$.
\end{proof}
\subsection{Hypothesis Testing Under Information Constraints with Sequential Interaction}\label{sec:infconst}
In this part we will study the distributed hypothesis testing problem under information constraints. 


\idnonidchannels*

\begin{proof}[Proof of \cref{id-nonid-channels}]
For channels $T^1,T^2,\dots,T^n$, define $P_{Y_1^n}$ (respectively, $Q_{Y_1^n}$) to be the output distributions under input distribution $P^{\otimes n}$ (respectively, $Q^{\otimes n}$). We have   \begin{align*}
        &P_{Y_1^n}(y_1,y_2,\dots,y_n)=\E_{X_1^n\sim P^{\otimes n}}\left[T^1(Y_1=y_1|X_1)T^2(Y_2=y_2|X_2,y_1)\dots T^n(Y_n=y_n|X_n,y_1^{n-1})\right],\\
        &Q_{Y_1^n}(y_1,y_2,\dots,y_n)=\E_{X_1^n\sim Q^{\otimes n}}\left[T^1(Y_1=y_1|X_1)T^2(Y_2=y_2|X_2,y_1)\dots T^n(Y_n=y_n|X_n,y_1^{n-1})\right],
    \end{align*}
For $\ls=\frac{\log(\frac{\alb}{\delta})}{\log(\frac{\alb}{\delta})+\log(\frac{\pi}{\delta})}$ define $\beta_*=\inf_{T\in\cT}\betls(TP,TQ)$. We use the following key lemma, proved in \Cref{lem2proof}:

    \begin{restatable}{lem}{idchannels}\label{idchannels}
        For any $(T^1,T^2,\dots,T^n)$ in the sequentially interactive setting, we have 
            $\betls(P_{Y_1^n},Q_{Y_1^n})
            \geq
            \beta_*^n.$
    \end{restatable}
     Let $\nsid =\nsid(P,Q,\pi,\delta,\cT)$. For some $t>0$, we choose $\epsilon\in(0,1)$ such that $(1+\epsilon)^{\nsid-1}\leq (\frac{\pi}{\delta})^{2{\ls} t}$. Note that it is possible because $\frac{\pi}{\delta}>1$. Now, let $T^*\in\cT$ be a channel such that $\betls(T^*P,T^*Q)\leq (1+\epsilon)\beta_*$. Based on the definition of $\nsid$ we have $e_\pi((T^*P)^{\otimes (\nsid-1)},(T^*Q)^{\otimes (\nsid-1)})>\delta.$ Using the upper bound on the error probability in \cref{fact: bayes}, we can write
    \begin{align*}
        &\delta\leq e_\pi\left((T^*P)^{\otimes (n_*-1)},(T^*Q)^{\otimes (n_*-1)}\right)
        \leq
        \pi^{\ls}\alb^{\lsb}\left(\betls(T^*P,T^*Q)\right)^{(n_*-1)}
        \leq 
        (\frac{\pi}{\delta})^{2\ls t}\pi^{\ls} \alb^{\lsb} \beta_*^{(n_*-1)}.
        \end{align*}
 This implies       
        $\beta_*^{n_*-1}
        \geq
        \frac{\delta}{\pi^{\ls}\alb^{\lsb}}(\frac{\delta}{\pi})^{2\ls t}
        =
        (\frac{\delta}{\pi})^{2\ls(1+t)}.$ 
     Suppose that for $\{T^1,T^2,\dots,T^n\}$ we have $e_\pi(P_{Y_1^n},Q_{Y_1^n})\leq\delta$. Then we can use \cref{thm: e_alpha_lb} to obtain a lower bound on $n$ as follows:
    \begin{align*}
        &\delta 
        \geq
        \frac{1}{4}\pi^{\ls} \alb^{\lsb} \left(\betls(P_{Y_1^n},Q_{Y_1^n})\right)^2
        \overset{(i)}{\geq}
        \frac{1}{4}\pi^{\ls} \alb^{\lsb} \beta_*^{2n} 
        \overset{(ii)}{\geq}
        \left(\frac{\delta}{\pi}\right)^{\ls}\pi^{\ls} \alb^{\lsb} \left(\frac{\delta}{\pi}\right)^{\frac{4\ls(1+t)n}{n_*-1}},
    \end{align*}
    which upon rearranging yields $n \geq \frac{1}{4(1+t)}(n_*-1).$ Above, $(i)$ follows from the lower bound on $\beta^*$ above and the fact that $\frac{1}{4}\geq\left(\frac{\delta}{\pi}\right)^{\ls}$. Letting $t\downarrow 0$ concludes the proof. 
\end{proof}

\section{Distributed Simple Hypothesis Testing Under Communication Constraint}\label{disthyp}
In light of \cref{id-nonid-channels}, it is evident that to study the sample complexity of distributed simple hypothesis testing with sequential interaction and constraints given by a set of channels $\cT$, it is enough to study $\inf_{T\in\cT}\betls(TP,TQ)$. 
In this section, we consider communication (or quantisation) constraints. We define $\cT_D$ as the set of all channels from $\cX$ to $\cY$ where $|\cY|=D$, i.e., each $X_i$ must be quantised to one of $D$ values before transmitting it to the central server. Without loss of generality, we assume $\cY = \{1, 2, \dots, D\}$. {\dgreen The result from~\cite[Corollary~3]{PenEtal24a} shows that the extreme points of the joint range of the output distributions induced by channels in $\cT_D$ are realised by threshold channels.}

\begin{define}[Threshold channels]\label{def: threshold}
We say $T\in\cT_D$ is a threshold channel if $T$ is characterised by disjoint intervals $\{I_1,I_2,\dots,I_D\}$ covering $\real$, where $T(x)=i$ if and only if $\frac{dP}{dQ}(x)\in I_i$. 
\end{define}

Observe that minimising the Hellinger affinity is equivalent to maximising the Hellinger divergence. 
 Standard data-processing results imply that 
$H_{\ls}(TP, TQ) \le H_{\ls}(P,Q).$
We require lower bounds on $\sup_{T \in \cT_D}H_{\ls}(TP, TQ)$ in terms of $H_{\ls}(P,Q)$; i.e., a ``reverse'' data-processing inequality. Such an inequality was shown in \cite{PenEtal23} for the Hellinger-$1/2$ divergence. Our main technical result of this section is a reverse data-processing inequality for a large class of $f$-divergences that we call \emph{TV-like} $f$-divergences. These are defined as follows: 
\begin{define}[TV-like $f$-divergence]\label{def: tv_f}
We say $D_{f}(P\parallel Q)\triangleq \E_Q[f(\frac{dP}{dQ})]$ is a \textbf{TV-like f-divergence} if $f(x)$ is a non-negative, convex function on $\real_+\cup\{+\infty\}$ with $f(1)=0$ which satisfies
\begin{enumerate}
    \item There exist $b,B\in \real_+$ such that $f$ is bounded by $B$ in $[0,b]$; i.e., for all $x\in[0,b]$, $f(x)\le B$.
    \item There exist $C_1,C_2\geq 0$ such that $f$ is almost linear on $(b, \infty)$; i.e., for all $x\in(b,\infty]$,  $C_1 x\leq f(x)\leq C_2 x$.
\end{enumerate}
\end{define}


    Some notable f-divergences that satisfy the definition above are $TV$, Hellinger-$1/2$, and indeed, all Hellinger-$\lambda$ divergences (c.f.\ \cref{Hlamnice}). Kullback--Leibler and $\chi^2$ divergences are not TV-like (because $f(x) = x\log x$ and $f(x) = x^2-1$ are superlinear). 

\setcounter{thr}{3}
\begin{thr}[Reverse data-processing inequality for TV-like divergences]\label{revdata}
    Let $P$ and $Q$ be two discrete distributions, and let $D_f$ be a TV-like f-divergence with parameters $b,B,C_1,C_2$. Then for any $D\geq 2$ there exists a threshold quantizer $T^*\in \cT_D$ such that the following inequality holds \begin{equation}
        D_f(T^*P\parallel T^*Q)
        \geq
        \frac{1}{52}\min\left\{1,\frac{26C_1}{C_2},\frac{D}{R}\right\}D_f(P\parallel Q), 
    \end{equation}
    where $R=1+\log(\frac{2B}{D_f(P\parallel Q)})$ and in the special case where $P$ and $Q$ are supported on points $k$, we can put $R=\min\{k,1+\log(\frac{2B}{D_f(P\parallel Q)})\}$. 
\setcounter{thr}{6}
\end{thr}

\begin{proof}[Proof of \cref{revdata}]
    We can write $D_f(P\parallel Q)$ as the sum of two terms 
    \begin{equation}
    D_f(P\parallel Q)
    =
    \underbrace{\E_Q\left[f\left(\frac{dP}{dQ}\right)\mathbbm{1}_{\{\frac{dP}{dQ}\leq b\}}\right]}_{S_1}
    +\underbrace{\E_Q\left[f\left(\frac{dP}{dQ}\right)\mathbbm{1}_{\{\frac{dP}{dQ}> b\}}\right]}_{S_2}
    \end{equation}
    Now either $S_1\geq\frac{1}{2}D_f(P\parallel Q)$ or $S_2\geq\frac{1}{2}D_f(P\parallel Q)$. We analyse these cases separately.\vspace{0.2cm}\\
    \textbf{Case 1.} Suppose $S_2\geq\frac{1}{2}D_f(P\parallel Q)$. Then we have
    \begin{equation}
        \frac{1}{2}D_f(P\parallel Q)
        \leq 
        \E_Q\left[f\left(\frac{dP}{dQ}\right)\mathbbm{1}_{\{\frac{dP}{dQ}> b\}}\right]
        \leq
        C_2 P\left[\frac{dP}{dQ}>b\right],
    \end{equation}
    where we used the fact that $f(x)\leq C_2 x$ for $x>b$.
    Now consider the two-level quantiser $T^*$ that maps any input $x\in\cX$ with $\frac{dP}{dQ}(x)>b$ to $a_1$ and maps the input to $a_2$ otherwise. For this quantiser,
    \begin{equation}
        D_f\left(T^*P\parallel T^*Q\right)
        \geq
        Q\left[\frac{dP}{dQ}>b\right]f\left(\frac{P\left[\frac{dP}{dQ}>b\right]}{Q\left[\frac{dP}{dQ}>b\right]}\right)
        \geq
        C_1 P\left[\frac{dP}{dQ}>b\right].
    \end{equation}
    For the last inequality, we used the fact that $f(x)\geq C_1 x$ for $x>b$. Combining the two bounds above, we obtain
    \begin{equation}
        D_f(T^*P\parallel T^*Q)
        \geq
        \frac{C_1}{2C_2}D_f(P\parallel Q).
    \end{equation}
    
    \textbf{Case 2.} Suppose $S_1\geq\frac{1}{2}D_f(P\parallel Q)$. Define the random variable $Z$ as $Z=f\left(\frac{dP}{dQ}(X)\right)\mathbbm{1}_{\{\frac{dP}{dQ}(X)\leq b\}}$ where $X \sim Q$. Note that $\E[Z]=S_1\geq\frac{1}{2}D_f(P\parallel Q)$. We now recall a reverse Markov inequality originally proved in \cite[Lemma 3.10]{PenEtal23}. 
        \begin{lem}[Reverse Markov inequality \cite{PenEtal23}]
    \label{lem:revMarkovD}
    Let $Z$ be a random variable over $[0, B)$ with expectation $\E[Z] > 0$. Let $R = 1 +  \log(B/\E [Y])$. Then
    \begin{align}
    \label{revMarkov}
    \sup_{0 \le \eta_1 \le  \cdots \le \eta_D = B} \sum_{j=1}^{D-1} \eta_{j} \prob\left( Z \in [\eta_{j}, \eta_{j+1}) \right) \geq \frac{1}{13} \E[Z] \min \left\{1, \frac{D}{R}  \right\}.
    \end{align}
    For the special case where $Z$ is supported on $k$ points, we may set $R = \min\{k,1 +  \log(B/\E [Y])\}$, and there is a $\poly(k,D)$ algorithm to find $\eta_j$'s that achieve the bound~\eqref{revMarkov}.
    \end{lem}
    
    Because $Z\in[0,B)$ we can use \cref{lem:revMarkovD} to get the following inequality
    \begin{align}
    \sup_{0\leq\eta_1\leq\dots\leq\eta_D=B}\sum_{i=1}^{D-1}\eta_i\prob[Z\in[\eta_i,\eta_{i+1})]
    \geq &
    \frac{1}{13}\min\left\{1,\frac{D}{1+\log(\frac{B}{\E[Z]})}\right\}\E[Z]
    \geq \nonumber\\&
    \frac{1}{26}\min\left\{1,\frac{D}{1+\log(\frac{2B}{D_f(P\parallel Q)})}\right\}D_f(P\parallel Q). \label{revmarkz}
    \end{align} 
    Consider any partition 
    $\{0\leq\eta_1\leq\dots\leq\eta_D=B\}$.  For each $[\eta_i,\eta_{i+1})$ define $\{I_i^j\}_{j=1}^t$ to be the preimage of $[\eta_i,\eta_{i+1})$ under the function $f$ i.e. $f^{-1}([\eta_i,\eta_{i+1}))$.  Note that $I_i^j$ is an interval and $t \le 2$, so, here we make the assumption that if $t=1$ we set $I_i^2=\emptyset$. \textcolor{black}{Note that just for ease of notation here we are using at most $2D$ quantization points, but we will replace $D$ by $\lfloor D/2\rfloor$, at the end to match the setting of theorem.} For $T^*$ that maps $x\rightarrow w^j_i$ when $\frac{dP}{dQ}(x)\in I^j_i$, we have
    \begin{align}
    \sum_{i=1}^{D-1}\eta_i\prob[Z\in[\eta_i,\eta_{i+1})]
    =&
    \sum_{i=1}^{D-1}\sum_{j=1}^{2} \eta_i Q\left[\frac{dP}{dQ}\in I_i^j\right]
    \leq \nonumber \\&
    \sum_{i=1}^{D-1}\sum_{j=1}^{2} Q\left[\frac{dP}{dQ}\in I_i^j\right]f\left(\frac{P[\frac{dP}{dQ}\in I_i^j]}{Q[\frac{dP}{dQ}\in I_i^j]}\right)
    =
    D_f(T^*P\parallel T^*Q). \label{lastineqdata}
    \end{align}
    Where the inequality follows from the fact that $\frac{P[\frac{dP}{dQ}\in I_i^j]}{Q[\frac{dP}{dQ}\in I_i^j]}\in I_i^j$, 
    so we have $f\left(\frac{P[\frac{dP}{dQ}\in I_i^j]}{Q[\frac{dP}{dQ}\in I_i^j]}\right)\geq \eta_i$. Without loss of generality we can assume that $D$ is even and then we get the final result from the combination of \cref{lastineqdata} and \cref{revmarkz}.
\end{proof}

{\dgreen
\paragraph{Constructing hypothesis tester}
When $\cX$ is finite, the proof of \Cref{revdata} is constructive. We first compute the likelihood ratios and use the algorithm of~\cite[Lemma~3.10]{PenEtal23} to construct the threshold quantiser $T^*$. Each agent applies $T^*$ to its observation, and the central server applies the MAP rule to the quantised observations. For fixed $D$, this construction is polynomial in $|\cX|$.
}

Our next theorem shows that the reverse data-processing inequality from~\Cref{revdata} cannot be improved in general; i.e., there exists a pair of distributions $(P, Q)$ such that the Hellinger-$\lambda$ divergence between $(TP, TQ)$ is necessarily reduced by a logarithmic factor. 
   {\dgreen
    \begin{thr}[Tightness of reverse data-processing for $H_\lambda$]\label{thm:datarevsharp}
    There are universal constants $c_1, c_2, C > 0$ such that the following holds. For every $\lambda \in (0,1)$ and every sufficiently large integer $k$, there exist $P$ and $Q$ supported on $2k$ points such that $H_\lambda(P,Q) \asymp k2^{-k},$ and for every $D \ge 2$,
    \begin{align*}
        \sup_{T \in \cTT_D} H_\lambda(TP, TQ) \leq C \min \left\{1, \frac{D}{k} \right\} H_\lambda(P,Q).
    \end{align*}
    Equivalently, 
    \begin{align*}
        \sup_{T \in \cTT_D} H_\lambda(TP, TQ) \leq C \min \left\{1, \frac{D}{\log(1/H_\lambda(P,Q))} \right\} H_\lambda(P,Q).
    \end{align*}
    \end{thr}
    \begin{proof}
    To prove~\Cref{thm:datarevsharp}, we generalise the strategy for proving the tightness of the reverse data-processing inequality for the case of $\lambda = 1/2$ from~\cite{PenEtal23}. First, observe that it is enough to prove the result for $\lambda \leq 1/2$ using the relation $H_\lambda(P,Q) = H_{\bar \lambda}(Q,P)$. 

Let $z_0$ be a sufficiently small constant, say $z_0 = 1/100$, and set $z_i = z_0/2^i$ for $i \in [k-1]$. Define the weights $w_i = r/z_i$, where $r$ is a normalising constant that ensures $\sum_{i=0}^{k-1} w_i = 1$. This means 
\begin{align*}
    r = z_0/(1 + 2 + \dots + 2^{k-1}) = \frac{z_0}{2^k-1} \asymp 2^{-k}.
\end{align*}
We now define the distributions $P$ and $Q$ supported on a set of size $2k$ as follows. First, for any $z_i$ identify the two solutions of the equation $f_\lambda(t) = z_i$. Note that exactly two solutions will exist since $f_\lambda$ is a convex function that satisfying $f_\lambda(0) = 1-\lambda > 1/2 > z_i$, and $f_\lambda$ monotonically decreases on $[0,1]$ and monotonically increases on $[1,\infty)$. Call these roots $a_i < 1 < b_i$. Define 
\begin{align*} 
    \theta_i = \frac{1 - a_i}{b_i-a_i}.
\end{align*}
Observe that $\theta_i$ has been chosen to satisfy $\bar \theta_i a_i + \theta_i b_i = 1 $ for all $i$. Set
\begin{align*}
    Q(i,+) = \theta_i w_i, \quad \text{ and } \quad Q(i,-) = \bar \theta_i w_i \quad \text{ for $0 \leq i \leq k-1$,}
\end{align*}
and 
\begin{align*}
    P(i,+) = \theta_i b_iw_i, \quad \text{ and } \quad P(i,-) = \bar \theta_i a_i w_i \quad \text{ for $0 \leq i \leq k-1$.}
\end{align*}
Observe that $Q$ is a valid distribution since $\sum_{i=0}^{k-1} Q(i,+)+Q(i,-) = \sum_{i=0}^{k-1} w_i = 1$. To check for $P$, observe that
\begin{align*}
    \sum_{i=0}^{k-1}\bigl(P(i,+)+P(i,-)\bigr)
    &= \sum_{i=0}^{k-1}\bigl(\theta_i b_i+\bar\theta_i a_i\bigr)w_i
    = \sum_{i=0}^{k-1} w_i = 1,
\end{align*}
where we used the identity $a_i\bar \theta_i + b_i \theta_i = 1$ for all $i$. Calculating $H_\lambda(P,Q),$ we obtain
\begin{align*}
    H_\lambda(P,Q) &= \sum_{i=0}^{k-1} Q(i,+) f_\lambda \left(\frac{P(i,+)}{Q(i,+)}\right) + Q(i,-) f_\lambda \left(\frac{P(i,-)}{Q(i,-)}\right)\\
    &= \sum_{i=0}^{k-1} Q(i,+) f_\lambda(b_i) + Q(i,-)f_\lambda(a_i)\\
    &= \sum_{i=0}^{k-1} Q(i,+) z_i + Q(i,-)z_i\\
    &= \sum_{i=0}^{k-1} w_iz_i\\
    &= rk \asymp 2^{-k}k.
\end{align*}
Taking logs, we get the equivalent relation $\log \frac{1}{H_\lambda(P,Q)} \asymp k$. The main technical step is the following result:
\begin{prop}\label{prop:one_block}
 Let $I \subseteq \{0,1,\dots, k-1\}$ be a nonempty contiguous block of indices. Define the upper and lower branch probabilities
 \begin{align*}
     Q_I^+ = \sum_{j \in I} w_j \theta_j, P_I^+ = \sum_{j \in I} w_j \theta_j b_j, Q_I^- = \sum_{j \in I} w_j\bar \theta_j, \text{ and } P_I^- = \sum_{j \in I} w_j \bar \theta_j a_j.
 \end{align*}
 Then
 \begin{align*}
     Q_I^+ f_\lambda\left( \frac{P_I^+}{Q_I^+} \right) \le Cr, \quad \text{ and } \quad Q_I^- f_\lambda\left( \frac{P_I^-}{Q_I^-} \right) \le Cr
 \end{align*}
 for some universal constant $C>0$.
\end{prop}
The proof of~\Cref{prop:one_block} is deferred to~\Cref{app:datarevsharp}. Using this, we may complete the proof as follows. Let $T$ be any deterministic threshold quantiser with at most $D$ outputs. Refine the threshold partition by adding a threshold at 1. This creates at most $D+1$ threshold cells. Refining the partition can only increase the output divergence, so it is enough to upper bound the $H_\lambda$ divergence after refinement. 

After refinement, the preimage of any output symbol lies entirely in the lower branch (i.e., $(j,-)$ for $j \in [0,k-1]$) or entirely in the upper branch (i.e., $(j,+)$ for $j \in [0,k-1]$). The monotonicity of $a_i$'s and $b_j$'s implies that every preimage corresponds to a contiguous block of indices in the upper or lower branches. \Cref{prop:one_block} implies the contribution of that block to the output divergences is at most $Cr$. Since there are at most $D+1$ such blocks, the total output divergence is at most $(D+1)Cr \leq C'Dr$. Combining this fact with the identity $H_\lambda(P,Q) = kr$ gives
\begin{align*}
    \sup_{T \in \cTT_D} H_\lambda(TP, TQ) \lesssim H_\lambda(P,Q) \frac{D}{k}.
\end{align*}
Combining the preceding bound with the data-processing inequality $H_\lambda(TP,TQ)\leq H_\lambda(P,Q)$ gives
\begin{align*}
    H_\lambda(TP,TQ)
    \lesssim
    \min\left\{1,\frac{D}{k}\right\}H_\lambda(P,Q).
\end{align*}
Finally, using the relation $k \asymp \log(1/H_\lambda(P,Q))$, we conclude the desired inequality 
\begin{align*}
    \sup_{T \in \cTT_D} H_\lambda(TP, TQ)
    \lesssim
    \min\left\{1,\frac{D}{\log(1/H_\lambda(P,Q))}\right\}H_\lambda(P,Q).
\end{align*}

\end{proof}
    }
In the following theorem we state our main result about the sample complexity of simple binary hypothesis testing under communication constraints.

\samplecompcommunication*

\begin{proof}[Proof of \cref{thm: sc_D}]
    We start with the proof of the upper bound. Our first lemma states that  $H_\lambda$ are TV-like divergences. The proof is a routine calculation and is deferred to \Cref{app:Hlamnice}.
    \begin{restatable}[Hellinger-$\lambda$ divergence is TV-like]{lem}{Hlamnice}\label{Hlamnice}
        For any $\lambda\in(0,1)$, $H_\lambda(P,Q)=1-\E[(\frac{dP}{dQ})^\lambda]$ is a TV-like f-divergence.
    \end{restatable}
    Now set $H_*\triangleq\sup_{T\in\cT_D}H_{\ls}(TP,TQ)$ for $\ls$ chosen as \cref{thm: e_alpha_lb} and let $\beta_*=1-H_*$. {\dgreen By data processing and \Cref{Hlambdaupper},
    \[
    H_*\leq H_{\ls}(P,Q)\leq 2\ls H_{1/2}(P,Q)\leq \frac{1}{2}.
    \]}
    From the argument in \Cref{subsec:compare} we know that $\nsseq(P,Q,\pi,\delta,\cT_D)\asymp \frac{\log(\frac{\pi}{\delta})}{\log(\frac{1}{\beta_*})}$ and because for $x\in(0,0.5)$ we have $e^{-2x}\leq 1-x\leq e^{-x}$ we conclude that $\nsseq(P,Q,\pi,\delta,\cT_D)\asymp \frac{\log(\frac{\pi}{\delta})}{H_*}$. {\dgreen If
    \[
    D\geq \log\left(\frac{1}{H_{\ls}(P,Q)}\right),
    \]
    then \Cref{revdata} gives $H_*\gtrsim H_{\ls}(P,Q)$, and the required upper bound follows directly. We may therefore assume that $D<\log(1/H_{\ls}(P,Q))$, in which case the logarithmic-loss calculation below applies.}
    We use the fact from \cref{Hlamnice} and \cref{revdata} to conclude that we have $H_*\gtrsim \frac{D}{\log\left(\frac{1}{H_{\ls}(P,Q)}\right)}H_{\ls}(P,Q)$ and we have
    \begin{align}
        \nsseq(P,Q,\pi,\delta,\cT_D)
        \asymp
        \frac{\log(\frac{\pi}{\delta})}{H_*}
        \lesssim
        \log\left(\frac{1}{H_{\ls}(P,Q)}\right)\frac{\log(\frac{\pi}{\delta})}{H_{\ls}(P,Q)}\frac{1}{D}
        \lesssim
        \frac{\log(n^*)}{D}n^*,
    \end{align}
    where in the final inequality we used that $n^* \gtrsim \frac{\log(\pi/\delta)}{H_{\ls}(P,Q)}$. So, $\log(1/H_{\ls}(P,Q)) \lesssim \log(n^*) - \log\log(\pi/\delta) \le \log(n^*).$ 

   To prove the lower bound, we use the optimality of threshold channels proved in \cite{PenEtal24a}. Specifically, \cite[Corollary 3]{PenEtal24a} shows that for any pair of distributions $P$ and $Q$ on $\cX$ where $\cX$ is a finite set, say $\lvert\cX\rvert=k$, {\dgreen threshold channels produce the extreme points of the convex set $\{(T(P),T(Q))\ s.t.\ T\in\cT_D\}$.} \color{black} So, to maximise a convex function on $\cT_D$, it is enough to maximise it on $\cT_{D}^{\mathrm{thresh}}$. \textcolor{black}{Then \cref{thm:datarevsharp} gives the lower bound for $\nsid$ and \cref{id-nonid-channels} transforms it to lower bound for $\nsnid$.} {\dgreen Consider $P$ and $Q$ as described in~\Cref{thm:datarevsharp} such that $H_*\lesssim\frac{D}{\log\left(\frac{1}{H_{\ls}(P,Q)}\right)}H_{\lambda}(P,Q)$. For this pair, we have
\begin{align*}
    n^*(P,Q,\pi,\delta,\cT_D)
    \asymp
    \frac{\log(\frac{\pi}{\delta})}{H_*}
    \gtrsim&
    \frac{\log(\frac{\pi}{\delta})}{DH_{\ls}(P,Q)}\log\left(\frac{1}{H_{\ls}(P,Q)}\right)
    \asymp
    \frac{n^*}{D}
    \log\left(\frac{1}{H_{\ls}(P,Q)}\right).
\end{align*}
}
\end{proof}

\section{Conclusion}
In this paper, we resolved two open problems related to the sample complexity of simple binary hypothesis testing. We showed that sequential interaction is of limited help, and we fully characterised the sample complexity under communication constraints. Our main technical tools, namely, the one-shot lower bound on Bayes error and the reverse data-processing inequality for $f$-divergences, are likely to have other applications as well, such as tightening Le Cam-style lower bounds in statistical estimation. There are several future research directions that can arise from this work. The question of whether the fully-interactive setting, also called the blackboard protocol (see \cite{KusNis96}) significantly reduces the sample complexity is still open. Finally, it would be interesting to see if the Bayes error lower bound from \cref{thm: e_alpha_lb} can be modified so that a lower bound can be computed without knowing the Bayes error in the first place. Such a bound could be tighter than Fano's information-theoretic lower bound and have numerous applications in statistical estimation and testing.

\section*{Acknowledgments}

We thank the anonymous reviewers of COLT 2025 for their helpful feedback. VJ is grateful for the hospitality of the Simons Institute at Berkeley where most  of this work was done.
\printbibliography

\appendix

\section{Proofs for  \Cref{one-shot}}
\subsection{Proof of  \Cref{techupbound}}\label{app: affinity-upperbound}

In this part, we prove \Cref{techupbound}, remember that we denote $\betl(\Ber(p),\Ber(q))$ by $\betl(p,q)$ in this lemma. also remember that we had
\begin{equation}
    \ls=\frac{\log(\frac{\alb}{\delta})}{\log(\frac{\alb}{\delta})+\log(\frac{\pi}{\delta})}
        \in[0.5,1).
\end{equation}
\techupbound*
\begin{proof}
    Suppose $e_\pi(p,q)=\delta \in (0,\pi)$. First, let us compute an explicit formula for $e_\pi(p,q)$.
    \begin{equation*}
    e_\pi(p,q)=\min\{\pi p,\alb q\} + \min\{\pi\pb,\alb\qb\}
    \end{equation*}
    Now if both minimums occur in the first argument of minimum or both occur in the second argument we will either have $\delta=\pi$ or $\delta=\alb$ which is not acceptable because both $\pi$ and $\alb$ are greater than $\delta$. So we have one of the following settings
    \begin{equation}\label{lines}
    \textbf{Line1:}\ \delta=\pi p +\alb(1-q)\quad\text{or}\quad \textbf{Line2:}\ \delta=\alb q+\pi(1-p).
    \end{equation}
    We define the set $L_1$(resp.\ $L_2$) to be all acceptable pairs $(p,q)$ on \textbf{Line1} (resp.\ \textbf{Line2}). One can see that we have the following forms for these two sets
    \begin{align}
        & L_1=\left\{(p,q)|\ q=\frac{\pi}{\alb}p+(1-\frac{\delta}{\alb})\ for\ p\in\left[0,\frac{\delta}{\pi}\right]\right\}\\
        &L_2=\left\{(p,q)|\ q=\frac{\pi}{\alb}p+\frac{\delta-\pi}{\alb}\ for\ p\in\left[1-\frac{\delta}{\pi},1\right]\right\}
    \end{align}
    It is also true that we can obtain $L_2$ from $L_1$ by substituting $(p,q)$ by $(\pb,\qb)$, essentially we have 
    $L_2=\{(\pb,\qb)\ s.t.\ (p,q)\in L_1\}$. Now the goal is to maximize $\betl(p,q)$ on $L_1\cup L_2$ but note that $\betl(p,q)=\betl(\pb,\qb)$, so without loss of generality we can do this maximization only on $L_1$. For points in $L_1$ it is easy to see that $p\in[0,\frac{\delta}{\pi}]$ and $q\in[1-\frac{\delta}{\alb},1]$. So we know that $\pb,q\neq 0$, and hence we can write 
    \begin{align}
        \max_{(p,q)\in L_1}\betl(p,q)
        &=
        \max_{(p,q)\in L_1} \left\{q\left(\frac{p}{q}\right)^{\lambda}+\Bar{p}\left(\frac{\Bar{q}}{\Bar{p}}\right)^{\lamb}\right\}
        \nonumber\\
        &\leq
        \max_{\substack{(p_0,q_0)\in L_1\\(p_1,q_1)\in L_1\\(p_2,q_2)\in L_1}}\left\{q_0\left(\frac{p_1}{q_1}\right)^{\lambda}+\Bar{p_0}\left(\frac{\Bar{q_2}}{\Bar{p_2}}\right)^{\lamb}\right\}
        \nonumber\\
        &=\label{maxbetl}
        \max_{(p_0,q_0)\in L_1}\left\{q_0 \left(\max_{(p_1,q_1)\in L_1}\frac{p_1}{q_1}\right)^{\lambda}+\Bar{p_0}\left(\max_{(p_2,q_2)\in L_1}\frac{\Bar{q_2}}{\Bar{p_2}}\right)^{\lamb}\right\}\,.
    \end{align}
    We can calculate the interior maximums explicitly as below
    \begin{align}
        \label{max1}
        & \max_{(p_1,q_1)\in L_1}\frac{p_1}{q_1}
        = 
        \max_{p_1\in [0,\frac{\delta}{\pi}]}\frac{p_1}{\frac{\pi}{\alb}p_1+1-\frac{\delta}{\alb}}
        = 
        \max_{p_1\in (0,\frac{\delta}{\pi}]}\frac{1}{\frac{\pi}{\alb}+(1-\frac{\delta}{\alb})\frac{1}{p_1}}
        =
        \frac{\delta}{\pi} \\&\label{max2}
        \max_{(p_2,q_2)\in L_1}\frac{\Bar{q_2}}{\Bar{p_2}}
        = 
        \max_{p_2\in[0,\frac{\delta}{\pi}]}\frac{\frac{\delta}{\alb}-\frac{\pi}{\alb}p_2}{\Bar{p_2}} 
        = 
        \max_{p_2\in[0,\frac{\delta}{\pi}]}\frac{\pi}{\alb}-\frac{\pi-\delta}{\alb}\frac{1}{\Bar{p_2}}
        =
        \frac{\delta}{\alb}
    \end{align}
    Here we choose $\lambda$ such that $(\frac{\delta}{\pi})^\lambda=(\frac{\delta}{\alb})^{\lamb}$ and it is easy to check that we should choose $\lambda=\ls=\frac{\log(\frac{\alb}{\delta})}{\log(\frac{\alb}{\delta})+\log(\frac{\pi}{\delta})}$. Now we just need to optimize on $(p_0,q_0)$. By substituting the values from \cref{max1} and \cref{max2} in \cref{maxbetl} we have
    \begin{align}
        \max_{(p,q)\in L_1}\betls(p,q) 
        \leq& 
        \left(\frac{\delta}{\pi}\right)^{\ls}\max_{(p_0,q_0))\in L_1}\{q_0+\Bar{p_0}\}\nonumber\\ 
        =&
        \left(\frac{\delta}{\pi}\right)^{\ls}\max_{p_0\in [0,\frac{\delta}{\pi}]}\left\{\left(2-\frac{\delta}{\alb}\right)-\left(1-\frac{\pi}{\alb}\right)p_0\right\} 
        \nonumber\\=& 
        \left(\frac{\delta}{\pi}\right)^{\ls}\left(2-\frac{\delta}{\alb}\right)\,,
    \end{align}
    which is stronger than the claimed bound.
\end{proof}

\subsection{Comparison of \cref{samplecomplxsimplecor} and the work of \cite{PenEtal24b}}\label{subsec:compare}

Here, our aim is to compare our result with that obtained by \cite{PenEtal24b}, so let us state their result in the following theorem.
    \begin{thr}[Theorem 2.1 of \cite{PenEtal24b}]\label{prevresult}
        Let $P$ and $Q$ be two arbitrary discrete distributions satisfying $H_{0.5}(p,q) \leq 0.125$.	%
    Let $\pi \in (0,1/2]$ be the prior parameter and $\delta \in
    (0,\pi/4]$ be the desired average error probability.
    \begin{enumerate}[noitemsep]
    \item (Linear error probability) If $\delta \in [\pi/100, \pi/4]$, then for $\lambda = \frac{(0.5 \log 2)}{\log(1/\pi)}$, we have
    \begin{align}
    \label{eq:main-thm-constant-failure}
    n^*(P,Q,\pi,\delta) 
    \asymp \frac{1}{\cH_{1 - \lambda} (p,q)}\,.
    \end{align}
    \item (Sublinear error probability) If $\delta \in \left[\pi^2, \frac{\pi}{100}\right]$, then for $T = \big\lfloor \frac{\log\left( \pi/\delta\right)}{ \log 8} \big\rfloor$ and $\pi' := \pi^{\frac{1}{T}}$, we have
    \begin{align}
    n^*(P,Q,\pi,\delta) &\asymp \log\left(   \pi/\delta\right) \cdot n^*(P,Q,\pi',\pi'/4) \asymp \frac{\log(\pi/\delta)}{H_{1-\lambda}(p,q)}, \text{ with $\lambda = \frac{0.5 \log 2}{\log(1/\pi')}$}
    \end{align}
    where the last expression can be evaluated using \Cref{eq:main-thm-constant-failure}.

    \item (Polynomial error probability) If $\delta \leq \pi^2$, then $n^*(P,Q,\pi,\delta) \asymp \frac{\log(1/\delta)}{H_{0.5}(p,q)}$.
    \end{enumerate}
    \end{thr}
    The first difference is that their result holds for all $\delta\in(0,\frac{1}{4}\pi]$ but we can not have $\delta=\pi/4$ in our result and if we want to cover a wider range of $\delta$ the coefficient $(1-c)$ gets close to 0 which is not desirable, so let us fix the range $\delta\in(0,\frac{1}{16}\pi]$.
    The next difference is that our result yields an upper bound and lower bound that differ by a factor of 4, whereas \Cref{prevresult} has much larger constants.
    
    Under the assumption that $H_{0.5}(P,Q) \leq 0.25$, the sample complexity of \Cref{samplecomplxsimplecor} becomes 
    \begin{align}
    \label{eq:this-paper-simplified}
        n^*(P,Q,\pi,\delta)\asymp \frac{ \log(\frac{\pi}{\delta})}{ \cH_{\lambda_*}(P,Q)} \text{  for $\ls=\frac{\log(\frac{\alb}{\delta})}{\log(\frac{\alb}{\delta})+\log(\frac{\pi}{\delta})}$} \,.
    \end{align}
    For completeness, we show this at the end of this section.

    We now verify that both \Cref{eq:this-paper-simplified} and \Cref{prevresult} yield same expressions up to constants.
    Observe that $\ls$ in \Cref{eq:this-paper-simplified} lies in $[0.5,1]$ and while the $\lambda$ values in \Cref{prevresult} lie in $[0,0.5]$. Thus,  we need to compare $1-\ls$ with the value of $\lambda$ provided in \cref{prevresult}. We check this fact for three regimes below
    \begin{itemize}
        \item (Linear error probability) Suppose $\delta=c\pi$ for some $c\in[1/100,1/16]$. The expressions (up to constants) given by \Cref{prevresult} and \Cref{eq:this-paper-simplified} are $\frac{1}{\cH_{1 - \lambda}(P,Q)}$ and $\frac{1}{\cH_{ \lambda^*}(P,Q)}$, respectively.
        Since $\lambda_* \geq 0.5$ and $\lambda \leq 0.5$, the desired claim follows from \Cref{Hlambdaupper} once we show $1 - \ls \asymp \lambda$.
        Indeed, we 
        have
        \begin{equation}
            1-\ls
            =
            \frac{\log(\frac{\pi}{\delta})}{\log(\frac{\alb}{\delta})+\log(\frac{\pi}{\delta})}
            =
            \frac{\log(\frac{1}{c})}{\log(\frac{\alb}{c^2\pi})}
            \asymp
            \frac{1}{\log(\frac{1}{\pi})}\,.
        \end{equation}
                
        \item (Sub-linear error probability) 
        Recall that in the second regime of the \cref{prevresult}, the sample complexity is within the constant factors of $\frac{\log(\pi/\delta)}{H_{1-\lambda}(p,q)}$ where $\lambda=\frac{0.5\log(2)}{\log(1/\pi')}$. 
        Comparing with \Cref{eq:this-paper-simplified},  it suffices to show that $\cH_{\lambda_*}(P,Q) \asymp \cH_{1 - \lambda}(P,Q)$.
        As argued above, it suffices to show $1 - \ls \asymp \lambda \asymp \frac{\log(\pi/\delta)}{\log(1/\pi)}$, where the last relation uses the value of $\pi'$. Now let $\delta=\pi^t$ for some $t\in [1,2]$. We have
        \begin{equation}
            1-\ls=
            \frac{\log(\frac{\pi}{\delta})}{\log(\frac{\pi\alb}{\delta^2})}
            =
            \frac{\log(\frac{\pi}{\delta})}{\log(\frac{\alb}{\pi^{2t-1}})}
            \asymp
            \frac{\log(\frac{\pi}{\delta})}{\log(\frac{1}{\pi})}\,.
        \end{equation}
        \item (Polynomial error probability)
        Now if $\delta=\pi^t$ for some $(t\geq 2)$, we can write
        \begin{equation}
            1-\ls=
            \frac{\log(\frac{\pi}{\delta})}{\log(\frac{\pi\alb}{\delta^2})}
            =
            \frac{\log(\frac{1}{\pi^{t-1}})}{\log(\frac{\alb}{\pi^{2t-1}})}
            \asymp
            1,.
        \end{equation}
        \end{itemize}
Hence, we conclude that the sample complexity bound that we derived is indeed within a constant factor of the sample complexity bounds derived in \cite{PenEtal24b}.

    \paragraph{Justifying \Cref{eq:this-paper-simplified}.} 
        We begin by imposing conditions that allows us to omit the ceil function from \Cref{samplecomplxsimplecor}.
        To do so, it is enough to ensure that we always have 
    \begin{equation}\label{regonbetl}
        0.5\ls\frac{\log(\frac{\pi}{\delta})}{\log\left(\frac{1}{\betls(P,Q)}\right)}\geq 1
    \end{equation}
    because then we can say that the sample complexity is always within the constant factors of $\frac{\log(\frac{\pi}{\delta})}{\log\frac{1}{\betl(P,Q)}}$. Since $\lambda_* \geq 0.5$ and $\delta \leq \pi/16$,
  it is enough to have $\betls(P,Q)\geq (\frac{1}{16})^{\frac{1}{4}}=0.5$ for \cref{regonbetl} to hold. Indeed,
    \begin{align*}
        \left(\frac{\delta}{\pi}\right)^{0.5\ls}
        \leq
        \left(\frac{1}{16}\right)^{\frac{1}{4}}
        \leq
        \betls\left(P,Q\right)
        \Rightarrow
        \frac{1}{\betls\left(P,Q\right)}
        \leq
        \left(\frac{\pi}{\delta}\right)^{0.5\ls}
        \Rightarrow
        \log\left(\frac{1}{\betls\left(P,Q\right)}\right)
        \leq
        0.5\ls\log\left(\frac{\pi}{\delta}\right)\,.
    \end{align*}
    Hence, it suffices to assume that $\betls(P,Q)\geq 0.5$ or equivalently $H_{\ls}(P,Q)\leq 0.5$. Also using \cref{Hlambdaupper}, we can see that it is enough to have $H_{0.5}(P,Q)\leq 0.25$ for this condition to hold. This is similar to the assumption in \Cref{prevresult} of $H_{0.5}\leq 0.125$. 
    Finally when $\betls(P,Q)\geq 0.5$, we have that $\log(1/\betls(P,Q) =  \log \left(\frac{1}{1- \cH_{\ls}(P,Q)}\right) \asymp \cH_{\ls}(P,Q)$ since $\cH_{\ls}(P,Q) \in [0,0.5]$ and for $x \in [0,0.5]$, we have that $\log(1/(1-x)) \asymp \log(1 + x) \asymp x$.
\subsection{Proof of \cref{idchannels}}\label{lem2proof}
Remember that we denoted $\beta_*\triangleq\inf_{T\in\cT}\betls(TP,TQ)$ and assumed that the channels $(T^1,T^2,\dots,T^n)$ are chosen in a sequentially interactive way, which means that channel $T^j$ can get as input $X_j$ and $Y_1^{j-1}=(Y_1,Y_2,\dots,Y_{j-1})$ but for any fixed choice of $y_1^{j-1}=(y_1,y_2,\dots,y_{j-1})$ the channel from $X_j$ is in the family of channels $\cT$. In this lemma, we denote the output probability of this sequence of channels with respect to the input probability $P^{\otimes n}$ (respectively, $Q^{\otimes n}$) by $P_{Y_1^n}$ (respectively, $Q_{Y_1^n}$). Here we state the proof of the following lemma.   
\idchannels*
\begin{proof}[Proof of  \cref{idchannels}]
    We prove this lemma by induction on $n$. For $n=1$, the inequality follows from the definition of $\beta_*$. So, suppose the statement is true for all $k\leq n-1$, we prove it for $n$. We can write the Hellinger affinity as below
    \begin{align}
    \betl(P_{Y_1^n},Q_{Y_1^n})
    =&
    \E_{Q_{Y_1^n}}\left[\left(\frac{P_{Y_1^n}}{Q_{Y_1^n}}\right)^\lambda\right]
    \nonumber\\
    &=
    \E_{Q_{Y^{n-1}_1}}\E_{Q_{Y_n|Y_1^{n-1}}}\left[\left(\frac{P_{Y_1^{n-1}}}{Q_{Y_1^{n-1}}}\right)^\lambda\left(\frac{P_{Y_n|Y_1^{n-1}}}{Q_{Y_n|Y_1^{n-1}}}\right)^\lambda\right]
    \nonumber\\&=
    \sum_{y_1^{n-1}\in\cY^{n-1}}Q_{Y_1^{n-1}}(y_1^{n-1})\left(\frac{P_{Y_1^{n-1}}(y_1^{n-1})}{Q_{Y_1^{n-1}}(y_1^{n-1})}\right)^\lambda\E_{Q_{Y_n|Y^{n-1}_1=y_1^{n-1}}}\left[\left(\frac{P_{Y_n|Y^{n-1}_1=y_1^{n-1}}}{Q_{Y_n|Y_1^{n-1}=y_1^{n-1}}}\right)^\lambda\right]\,.\label{betlexpansion}
    \end{align}
    Now for some $\epsilon>0$, let $\bar{y}_1^{n-1}\in\cY^{n-1}$ be such that 
    \begin{equation}\label{eq:inf-affinity}
    (1-\epsilon)\E_{Q_{Y_n|Y_1^{n-1}=\bar{y}_1^{n-1}}}\left[\left(\frac{P_{Y_n|Y_1^{n-1}=\bar{y}_1^{n-1}}}{Q_{Y_n|Y^{n-1}_1=\bar{y}_1^{n-1}}}\right)^\lambda\right]
    \leq
    \inf_{y^{n-1}_1\in\mathcal{Y}^{n-1}}\E_{Q_{Y_n|Y_1^{n-1}=y_1^{n-1}}}\left[\left(\frac{P_{Y_n|Y_1^{n-1}=y_1^{n-1}}}{Q_{Y_n|Y_1^{n-1}=y_1^{n-1}}}\right)^\lambda\right]\,.
    \end{equation}
    Then by using the induction assumption and \cref{eq:inf-affinity} we can lower bound  \cref{betlexpansion} 
    \begin{align}
        \betl(P_{Y_1^n},Q_{Y_1^n})
        &=
        \sum_{y_1^{n-1}\in\cY^{n-1}}Q_{Y_1^{n-1}}(y_1^{n-1})\left(\frac{P_{Y_1^{n-1}}(y_1^{n-1})}{Q_{Y_1^{n-1}}(y_1^{n-1})}\right)^\lambda\E_{Q_{Y_n|Y^{n-1}_1=y_1^{n-1}}}\left[\left(\frac{P_{Y_n|Y^{n-1}_1=y_1^{n-1}}}{Q_{Y_n|Y_1^{n-1}=y_1^{n-1}}}\right)^\lambda\right]\nonumber
        \\&\geq
        (1-\epsilon)\E_{Q_{Y_n|Y_1^{n-1}=\bar{y}_1^{n-1}}}\left[\left(\frac{P_{Y_n|Y_1^{n-1}=\bar{y}_1^{n-1}}}{Q_{Y_n|Y^{n-1}_1=\bar{y}_1^{n-1}}}\right)^\lambda\right]\beta_*^{n-1}
\nonumber        \\&\geq
        (1-\epsilon)\beta_*^{n},
    \end{align}
    where we used the assumption of induction, \cref{eq:inf-affinity} for the first inequality, and the definition of $\beta_*$ for the last inequality. Now, it is enough to let $\epsilon\downarrow 0$ to complete the proof.
\end{proof}

\section{Proofs for  \Cref{disthyp}}
\subsection{Proof of \cref{Hlamnice}}\label{app:Hlamnice}
\Hlamnice*
\begin{proof}[Proof of \cref{Hlamnice}]
{\dgreen
It is enough to prove the required TV-like bounds for $\lambda\in(0,1/2]$. Indeed, when $\lambda\in(1/2,1)$, we have $\Bar{\lambda}=1-\lambda\in(0,1/2)$. Using the symmetry
\[
H_\lambda(P,Q)=H_{\Bar{\lambda}}(Q,P),
\]
we may exchange the roles of $P$ and $Q$ and apply the result with parameter $\Bar{\lambda}$. We therefore assume throughout the proof that $\lambda\in(0,1/2]$.
}
Let $f_\lambda(x) = \lamb + \lambda x - x^\lambda$. One can verify that $H_{\lambda}(P,Q)=\E[f(\frac{dP}{dQ})]$. We prove the following properties for $f_{\lambda}(x)$.
    \begin{enumerate}[label=\roman*.]

\item \emph{Non-negativity}: $f_\lambda(x) \ge 0$ for all $x \in [0, \infty)$, and $f_\lambda(1) = 0$. This fact follows from the AM-GM inequality.

\item \emph{Monotonicity}: $f_\lambda(\cdot)$ is  decreasing on $[0,1]$ and increasing on $[1, \infty)$. This is easily seen by computing the derivative, \begin{align*}
f_\lambda'(x) = \lambda - \lambda x^{\lambda - 1},
\end{align*}
and checking its sign.

\item \emph{Boundedness over an interval}: $f_\lambda(x) \le 1$ for $x \in [0, 1/\lambda]$. As $f_\lambda(0) = \bar \lambda < 1$, the monotonicity property implies it is enough to check $f_\lambda(1/\lambda) \le 1$. Plugging in the value of $x$ to be $1/\lambda$,
\begin{align*}
f_\lambda(1/\lambda) &= \bar \lambda + 1 - (1/\lambda^\lambda).
\end{align*}
Differentiating with respect to $\lambda$,
\begin{align*}
\frac{d}{d\lambda} f_\lambda(1/\lambda) &= -1 + \lambda^{-\lambda}(1+\log \lambda) \le 0,
\end{align*}
where the final inequality uses $\lambda^\lambda \ge \lambda \ge 1 + \log \lambda$. Thus, $f_\lambda(1/\lambda)$ is maximized in the limit as $\lambda \to 0$. As $\lambda \to 0$, it can easily be verified that $f_\lambda(1/\lambda)$ tends to 1. Hence, we conclude that $f_\lambda(1/\lambda) \le 1$. 

\item \emph{Eventual linearity}: For $x \ge 1/\lambda$, we claim that $f_\lambda(x) \asymp \lambda x$. Specifically,  
\begin{align*}
c_1 \lambda x \le f_\lambda(x) \le c_2 \lambda x,
\end{align*}
where $c_1 = (1-\sqrt{2/3})/4$ and $c_2 = 2$. The upper bound is almost immediate:
\begin{align*} 
f_\lambda(x) \le \bar \lambda + \lambda x \le 2\lambda x,
\end{align*}
where we used $\bar \lambda < 1 \le x\lambda$ when $x \ge 1/\lambda$.

To prove the lower bound, we use the convexity of $f_\lambda(\cdot)$ to write
\begin{align*}
f_\lambda(x) &\ge (x - 3/(4\lambda))f_\lambda'(3/(4\lambda)) \ge x f_\lambda'(3/(4\lambda))/4.
\end{align*}
Above, we used that since $x \ge 1/\lambda$, $(x - 3/(4\lambda)) \ge x - 3x/4 = x/4.$ Evaluating $f_\lambda'(3/(4\lambda))$,
\begin{align*}
f_\lambda'(3/(4\lambda)) &= \lambda \left( 1 - (4\lambda/3)^{\bar \lambda}\right).
\end{align*}
It is not hard to check that $\left( 1 - (4\lambda/3)^{\bar \lambda}\right)$ is a decreasing function of $\lambda$ over $[0,1/2]$, and hence
\begin{align*}
f_\lambda'(3/(4\lambda) \ge \lambda (1-\sqrt{2/3}).
\end{align*}
Plugging in, we conclude the result.
\end{enumerate}
\end{proof}
\subsection{Proof of~\Cref{prop:one_block}}\label{app:datarevsharp}
{\dgreen
The proof of this proposition requires two lemmas.
\begin{lem}\label{lem:f_approx}
    There are universal constants $c, C >0$ such that for every $0 < \lambda \le 1/2$ and every $u \ge 0$ and every $v \in [0,1/2]$, 
    \begin{enumerate}
        \item[(1)] $c\lambda \min\{u^2, u\} \leq f_\lambda(1+u) \leq C\lambda \min\left\{u^2, u \right\}$

        \item[(2)] $c\lambda v^2 \leq f_\lambda(1-v) \leq C\lambda v^2$
    \end{enumerate}
\end{lem}
\begin{proof}
    We prove the upper bounds in (1) first. Note that $f''_\lambda(t) = \lambda\bar\lambda t^{\lambda-2}$. For $t \geq 1$, we get $f''_\lambda(t) \leq \lambda\bar\lambda \leq \lambda$. Write 
    \begin{align*}
        f_\lambda(1+u) = \int_{0}^u f'_\lambda(1+t)dt = \int_0^u \int_0^t f''_\lambda(1+s)ds dt \leq \lambda u^2/2.
    \end{align*} 
    To get the linear upper bound, observe that convexity of $f_\lambda$ gives that for all $t \geq 1$,
    \begin{align*}
        f_\lambda(1+u) \leq \frac{f_\lambda(1+tu)}{t}.
    \end{align*}
    Taking the limit as $t \to \infty$ gives $f_\lambda(1+u) \leq \lambda u$. Thus, $f_\lambda(1+u) \leq \min\{ \lambda u^2/2, \lambda u\}$.

    To prove the lower bound in (1), consider two cases: $u \leq 1$ and $u > 1$. When $u \leq 1$, we use that for $x \in [0,1]$, the double derivative satisfies $f''_\lambda(1+x) \geq f_\lambda''(2) \geq \lambda/8$. Using a similar argument as above yields $f_\lambda(1+u) \geq \lambda u^2/16.$ {\dgreen For $u \geq 1$, convexity gives $f_\lambda(1+u) \geq uf_\lambda(2) \geq \lambda u/16$, where we use $f_\lambda(2) \geq \lambda/16$ from the previous expression.}

    Both the lower and upper bounds in (2) can be shown by observing that for $x \in [0,1/2]$, the double derivative satisfies $f_\lambda''(1) \leq f_\lambda''(1-v) \leq f_\lambda''(1/2)$, and both $f_\lambda''(1) \geq \lambda/2$ and $f_\lambda''(1/2) \leq 4\lambda$.
\end{proof}

\begin{lem}\label{lem:roots_approx}
    Let $z_0 = 1/100$. For each $\lambda \in (0, 1/2]$ and each $z \in (0, z_0]$, let $a(z) < 1 < b(z)$ be the two roots of $f_\lambda(t) = z$. Write $a(z) = 1-v(z)$ and $b(z) = 1+u(z)$. Set $\theta(z) = \frac{1-a(z)}{b(z)-a(z)} = \frac{v(z)}{u(z)+v(z)}$. Then there are universal constants $c, c^*, C >0$ such that the following hold for all $\lambda \in (0,1/2]$ and all $z \in (0, z_0]$:
    \begin{enumerate}
        \item[(1)] $\theta(z) \leq 1/2$
        \item[(2)] $\theta(z)u(z) \leq C\min\{\sqrt{\frac{z}{\lambda}}, 1 \}$
        \item[(3)] If $z < c^*\lambda$, then $\theta(z) \geq c, u(z) \leq C\sqrt{\frac{z}{\lambda}}$, and $v(z) \leq C\sqrt{\frac{z}{\lambda}}$.
        
        \item[(4)] If $z \ge c^*\lambda$, then $\theta(z) \geq c\frac{\lambda}{z}.$
    \end{enumerate}
\end{lem}
\begin{proof}
Proof of (1): For $v \in (0,1]$, set $G(v):= f_\lambda(1-v) - f_\lambda(1+v) = (1+v)^\lambda - (1-v)^\lambda - 2\lambda v$. Then $G(0)=0$ and $G'(v) = \lambda\left[(1+v)^{\lambda-1} + (1-v)^{\lambda-1} - 2 \right]$. By convexity of $t \mapsto t^{\lambda-1}$, we have $(1+v)^{\lambda-1} + (1-v)^{\lambda-1} \geq 2$, that is, $G'(v) \geq 0$. This means $G(v) \geq 0$, and hence $f_\lambda(1-v) \ge f_\lambda(1+v).$ Now $f_\lambda$ is decreasing on $(0,1)$ and increasing on $(1,\infty)$. Since $f_\lambda(1-v(z)) = f_\lambda(1+u(z))$, the above comparison implies $u(z) \geq v(z)$, which gives $\theta(z) \leq 1/2$.

Proof of (2): If $z > \lambda$, then $\theta(z)u(z) \leq v(z) \leq 1$. Thus, it is enough to consider $z \leq \lambda$. We have two cases to consider: either $u \leq 1$ or $u > 1$. If $u \leq 1$, then~\Cref{lem:f_approx} yields  $f_\lambda(1+u) = z \geq c\lambda u^2$, which yields the bound $u \lesssim \sqrt{z/\lambda}$. If $u > 1$, then $z \geq c\lambda u$, and in particular $z \geq c\lambda$. This combined with $z \leq \lambda$ gives $u \leq 1/c$; since $\sqrt{z/\lambda} \geq \sqrt{c}$, the same bound follows after changing $C$.

Proof of (3): Using \Cref{lem:f_approx}, choose $c^*$ small enough such that if $z < c^* \lambda$ then $v(z) < 1/2$. For this $c^*$, we have $c\lambda v^2 \leq z \le C\lambda v^2$ and so  $v \asymp \sqrt{z/\lambda}$. Similarly, from the proof of (2) above~\Cref{lem:f_approx}, we conclude {\dgreen $u \asymp \sqrt{z/\lambda}$}. Since both $u$ and $v$ are of the same order, $\theta \asymp 1$.

Proof of (4): When $z \geq c^*\lambda$, we have $v \geq c$ for some constant $c$. This is because either $v > 1/2$, in which case we are done, or $v \le 1/2$, in which case the quadratic approximation from~\Cref{lem:f_approx} provides the lower bound. Using the lower bound from \Cref{lem:f_approx}, if $u \geq 1$ then $z = f_\lambda(1+u) \geq c\lambda u$, giving $u \leq Cz/\lambda$ for a suitable $C$. If $u < 1$, then it is easy to see that $u \leq 1 \leq z/(c^*\lambda) = Cz/\lambda$ for a suitable $C$. Combining these bounds, we conclude that $\theta = v/(u+v) \asymp v/u \gtrsim \lambda/z.$
\end{proof}

We can now proceed with the proof of~\Cref{prop:one_block}. Let $s := \min_{j \in I} z_j$ and $W_I := \sum_{j \in I} w_j.$ Note that $Q_I^+ \leq W_I$ and $Q_I^- \leq W_I$. Moreover, the dyadic construction of the $w_j$'s gives $r/s \leq W_I \leq 2r/s$. Thus, it is enough to prove
\begin{align*}
    f_\lambda\left(\frac{P_I^+}{Q_I^+}\right) \lesssim s, \quad \text{ and } \quad f_\lambda\left(\frac{P_I^-}{Q_I^-}\right) \lesssim s.
\end{align*}
We deal with the positive branch first. Let $\frac{P_I^+}{Q_I^+} = 1 + U_I$, where $U_I = \frac{\sum_{j \in I} z_j^{-1}\theta_j u_j}{\sum_{j \in I}z_j^{-1}\theta_j}$. Let $c^*$ be the universal constant from \Cref{lem:roots_approx}. We split into two cases:
\paragraph{Upper branch case 1 ($s < c^*\lambda$):} By~\Cref{lem:roots_approx} (3), we have $\theta(s) \ge c$, so $\sum_{j \in I}z_j^{-1}\theta_j \gtrsim 1/s$. For the numerator, \Cref{lem:roots_approx} (2) implies
\begin{align*}
    \sum_{j \in I} z_j^{-1}\theta_j u_j &\leq \sum_{j \in I: z_j \leq \lambda} \frac{1}{z_j}{\dgreen\sqrt{\frac{z_j}{\lambda}}} + \sum_{j \in I: z_j > \lambda} \frac{1}{z_j}\\
    &\le \frac{C}{\sqrt{\lambda s}} + \frac{C}{\lambda}\\
    &\le \frac{C}{\sqrt{\lambda s}},
\end{align*}
where we used the dyadic construction in the first inequality to only retain the smallest relevant $z_j$ in each sum, and the fact that $z \le c^*\lambda$ in the second inequality. Combining, we get $U_I \leq C\sqrt{s/\lambda}$,  and the quadratic upper bound from \Cref{lem:f_approx} gives $f_\lambda(1+U_I) \lesssim \lambda U_I^2 \lesssim s$, where we used the monotonicity of $f_\lambda$ on the positive branch.

\paragraph{Upper branch case 2 ($s \ge c^*\lambda$):} \Cref{lem:roots_approx} (4) gives $\theta(s) \geq c\lambda/s$, and so $\sum_{j \in I} z_j^{-1}\theta_j \ge c\lambda/s^2$. For the numerator, use the simple bound $\theta_j u_j = u_jv_j/(u_j+v_j) \leq v_j \leq 1$ to arrive at the upper bound $\sum_{j \in I} z_j^{-1}\theta_j u_j \leq C/s$, where we used the dyadic construction once more. Combining these bounds yields $U_I \lesssim s/\lambda$, and the linear upper bound from \Cref{lem:f_approx} leads to the desired result $f_\lambda(1+U_I) \lesssim s$.

We now consider the lower branch. As before, let $P_I^-/Q_I^- = 1 - V_I$, where $V_I = \frac{\sum_{j \in I} z_j^{-1}\bar \theta_j v_j}{\sum_{j \in I} z_j^{-1}\bar\theta_j}$. We consider two cases again.

\paragraph{Lower branch case 1 ($s < c^*\lambda$):} By \Cref{lem:roots_approx} (1), we have $1/2 \leq \bar \theta_j \leq 1$. Hence, the denominator is easily seen to satisfy $2/s \geq \sum_{j\in I} z_j^{-1}\bar\theta_j \geq 1/(2s)$, where we used the dyadic construction for the upper bound. For the numerator, we use {\dgreen $v(z) \leq 2\theta(z)u(z)$} and $\theta(z)u(z) \leq C\min\left\{\sqrt{\frac{z}{\lambda}}, 1\right\}$, by 
\Cref{lem:roots_approx} (2). Thus,
\begin{align*}
    \sum_{j \in I} z_j^{-1}\bar\theta_j v_j &\leq \sum_{j \in I: z_j \le \lambda} \frac{1}{z_j} \sqrt{\frac{z_j}{\lambda}} +  \sum_{j \in I: z_j > \lambda} \frac{1}{z_j}\\
    &\leq \frac{C}{\sqrt{\lambda s}} + \frac{C}{\lambda}\\
    &\leq \frac{C}{\sqrt{\lambda s}},
\end{align*}
where we used the dyadic construction in the first inequality to retain the smallest $z_j$, and $s < c^*\lambda$ in the second inequality. Hence, $V_I \leq C\sqrt{s/\lambda}$. We may choose $c^*$ small enough so that the right hand side is at most $1/2$. Using the quadratic approximation from \Cref{lem:f_approx} and the monotonicity of $f_\lambda$, we obtain $f_\lambda(1-V_I) \leq C\lambda V_I^2 \leq Cs$.

\paragraph{Lower branch case 2 ($s \ge c^*\lambda$:)} From case 1 above, the denonimator is $\asymp 1/s$. The monotonicity of the $a_j$'s in $z_j$, gives $a_j \leq a(s)$. Hence, we may write
\begin{align*}
    \frac{P_I^-}{Q_I^-} \geq \frac{a(s)/(2s)}{2/s} = a(s)/4.
\end{align*}
Since $f_\lambda$ is decreasing over $(0,1)$, 
\begin{align*}
    f_\lambda\left(\frac{P_I^-}{Q_I^-}\right) \leq f_\lambda(a(s)/4).
\end{align*}
Comparing $f_\lambda(a/4)$ with $f_\lambda(a)$ for $a \in (0,1)$ yields
\begin{align*}
    f_\lambda(a/4) - f_\lambda(a) = -3\lambda a/4 + a^\lambda(1-(1/4)^\lambda) \leq (1-(1/4)^\lambda) \leq \lambda \log 4.
\end{align*}
Hence, we get {\dgreen $f_\lambda(a(s)/4) \leq s + C\lambda$}. But for the case under consideration, $s\geq c^*\lambda$, and so $f_\lambda(a(s)/4) \leq Cs$.

This completes the proof of \Cref{prop:one_block}.
}
\end{document}